\documentclass[reqno]{amsart}
\usepackage{amssymb}

\numberwithin{equation}{section}
\newtheorem{theorem}{Theorem}[section]
\newtheorem{proposition}[theorem]{Proposition}
\newtheorem{lemma}[theorem]{Lemma}

\DeclareMathOperator{\supp}{supp}
\DeclareMathOperator{\tr}{tr}

\DeclareMathOperator{\esssupp}{ess \, supp}

\DeclareMathOperator{\essinf}{ess \, inf}

\renewcommand\H{\mathcal{H}}

\renewcommand\L{\mathrm{L}}
\newcommand\R{\mathbb R}
\newcommand\N{\mathbb N}
\newcommand\C{\mathbb C}
\newcommand\Z{\mathbb Z}

\newcommand\di{\mathrm d}

\newcommand\cF{\mathcal{F}}

\newcommand\e{\mathrm{e}}
\newcommand{\la}{\langle}
\newcommand{\ra}{\rangle}
\renewcommand\P{\mathbb P}
\newcommand\E{\mathbb E}

\newcommand\cE{\mathcal{E}}
\newcommand\cI{\mathcal{I}}

\newcommand\eps{\varepsilon}
\newcommand\vphi{\varphi}
\renewcommand{\d}{\mathrm{d}}

\newcommand{\pr}{\prime}

\newcommand{\bom}{{\boldsymbol{{\omega}}}}
\newcommand{\btau}{{\boldsymbol{{\tau}}}}

\newcommand\beq{\begin{equation}}
\newcommand\eeq{\end{equation}}

\newcommand{\abs}[1]{\left\lvert #1 \right\rvert}
\newcommand{\norm}[1]{\left\lVert #1 \right\rVert}
\newcommand{\scal}[1]{\left\langle #1 \right\rangle}
\newcommand{\set}[1]{\left\{ #1 \right\}}
\newcommand{\pa}[1]{\left( #1 \right)}

\newcommand{\hnorm}[1]{\left\{ \!\left\{ #1\right\}\! \right\}}

\newcommand{\eq}[1]{\eqref{#1}}

\newcommand{\up}[1]{^{(#1)}}

\begin{document}

\title[Poisson Statistics for Eigenvalues]
{Poisson Statistics for Eigenvalues of  Continuum  Random Schr\"odinger Operators}

\author[J.-M. Combes]{Jean-Michel Combes}
\address[Combes]{Universit\'e du Sud: Toulon et le Var, D\'epartement de Math\'ematiques,
F-83130 La Garde, France}
\email{combes@cpt.univ-mrs.fr}

\author[F. Germinet]{Fran\c cois Germinet}
\address[Germinet]{Universit\'e de Cergy-Pontoise,
CNRS UMR 8088, IUF, D\'epartement de Math\'ematiques,
F-95000 Cergy-Pontoise, France}
\email{germinet@math.u-cergy.fr}

\author[A. Klein]{Abel Klein}
\address[Klein]{University of California, Irvine,
Department of Mathematics,
Irvine, CA 92697-3875,  USA}
 \email{aklein@uci.edu}

\thanks{A.K was  supported in part by NSF Grant DMS-0457474.}


\begin{abstract}
We show absence of energy levels repulsion for the eigenvalues of    random Schr\"odinger operators in the continuum.  We prove that, in the localization region at the bottom of the spectrum,  the properly rescaled eigenvalues of a continuum  Anderson Hamiltonian are  distributed as a Poisson point process with intensity measure given by the density of states.  We also obtain simplicity of the eigenvalues. We derive a Minami estimate for continuum  Anderson Hamiltonians.  We also give a simple and transparent proof of Minami's estimate for the (discrete) Anderson model.
\end{abstract}

\maketitle


\section{Introduction}

In this article we show absence of energy levels repulsion for the eigenvalues of    random Schr\"odinger operators in the continuum. 
We prove that, in the localization region at the bottom of the spectrum,  the properly rescaled eigenvalues of a continuum  Anderson Hamiltonian are  distributed as a Poisson point process with intensity measure given by the density of states. We also obtain simplicity of the eigenvalues in that region.

Local fluctuations of eigenvalues of   random  operators is believed to distinguish between  localized and delocalized regimes, indicating an Anderson metal-insulator transition. Exponential decay of eigenfunctions implies that disjoint regions of space are uncorrelated and  create almost independent  eigenvalues, and thus absence of energy  levels repulsion, which is mathematically translated in terms of a Poisson point process. On the other hand, extended states imply that distant regions  have mutual influence, and thus create some repulsion between energy levels. Local fluctuations of eigenvalues have been studied within the context of random matrix theory, in particular Wigner matrices and GUE matrices,  cf.\  \cite{B,DPS,ESY1,ESY2,J1,J2,SS}  and references therein.  It is challenging  to understand random hermitian band matrices from the perspective of their eigenvalues fluctuations, by  proving  a transition between Poisson statistics and a semi-circle law for the density of states (a signature of energy   levels repulsion), and relate this to the (discrete) Anderson model, cf.\   \cite{B,DPS}.  CMV matrices are another class of random matrices for which Poisson statistics     and a transition to energy   levels repulsion  have been proved \cite{KS,St1,St2}.

For random Schr\"odinger operators, Poisson statistics for eigenvalues  was first proved by Molchanov \cite{Mo} for the same one-dimensional continuum random Schr\"o\-dinger operator for which Anderson localization was first rigorously established \cite{GMP}.  Molchanov's proof was based on a detailed analysis of localization in finite intervals  for this particular random Schr\"odinger operator \cite{Mo1}.

Poisson statistics for eigenvalues of the Anderson model was established by Minami \cite{Mi}.  The Anderson model, a  random Schr\"odinger operator on $\ell^2(\Z^d)$, is the discrete analogue of the Anderson Hamiltonian. A crucial ingredient in Minami's proof  is an estimate of the probability of two or more eigenvalues in an interval.  The  key step in the proof of this estimate, namely \cite[Lemma~2]{Mi},   estimates the average of  a determinant whose entries are matrix elements of the imaginary part of the resolvent. The more recent proofs of Minami's estimate by Bellissard,  Hislop and Stolz \cite{BHS} and    Graf and Vaghi \cite{GV} are variants of Minami's. Since those arguments do not seem to extend to the continuum, up to now a Minami-type estimate and Poisson statistics for the eigenvalues have been challenging  questions for continuum  Anderson Hamiltonians.

In this article we introduce a totally  new approach to Minami's estimate.  Unlike the previous  approach, ours   relies on averaging spectral projections, a technique that  does extend to the continuum.  Combined with  a  property of rank one perturbations,  it provides a simple and transparent  proof of Minami's estimate for the Anderson model, valid for  single-site probability distributions with compact support and no atoms, which is presented here as an illustration of the method.  On the continuum, our proof of Minami's estimate  circumvents the unavailability of that rank one property  by averaging   the spectral shift function, using   refined bounds on the density of states not previously available.

Once we have Minami's estimate in the continuum, we prove Poisson statistics for eigenvalues  of the Anderson Hamiltonian.  We start by  approximating the point process defined by the rescaled  eigenvalues    by  superpositions of independent point processes, as in \cite{Mo,Mi}.  But our proof that these superpositions converge weakly to the desired Poisson point process differs from Minami's  for the Anderson model, since his way of identifying the intensity measure of the Poisson process, which relies on complex analysis,  is not readily  applicable in the continuum.  We identify this intensity measure using methods of  real analysis.

 Klein and Molchanov \cite{KM} showed that Minami's estimate  implies  simplicity of eigenvalues for the Anderson model, a result previously obtained by Simon \cite{Si} by different methods. Their arguments can also be applied in the continuum, so we also obtain simplicity of eigenvalues in the continuum. Previous results \cite{CH,GKsudec}  proved only finite multiplicity of the eigenvalues in the localization region.

  \section{Main results}
  
To state our results we introduce the following notation.  We write
\beq \label{box}
\Lambda_{L}(x):= x +\left[-\tfrac L 2, \tfrac L 2\right[^d
\eeq
for the  (half open-half closed)
box of side $L>0$ centered at $x\in \R^d$.  By $\Lambda_L$ we denote a box $\Lambda_{L}(x)$ for some $x \in \R^d$.  Given a box $\Lambda=\Lambda_{L}(x)$, we 
set $\widetilde{\Lambda}= \Lambda \cap \Z^d$. If $B$ is a set, we write  $\chi_B$ for its characteristic function. 
 We set $\chi^{(L)}_x:=\chi_{\Lambda_L(x)}$.  The  Lebesgue measure of a  Borel set $B \subset \R$ will be denoted  by $\abs{B}$.  If $r>0$, we denote by $[r]$ the largest integer less than equal to $r$, and by $[[r]]$ the smallest integer bigger than $r$. By a constant we will always mean a finite constant. Constants such as $C_{a,b,\ldots}$  will be finite and depending only on the parameters or quantities $a,b,\ldots$; they will be independent of other  parameters or quantities in the equation.

 We consider  random Schr\"odinger 
operators on 
$\mathrm{L}^2(\mathbb{R}^d)$ of the type
\beq\label{AndH}
H_{\bom}: =  -\Delta + V_{\mathrm{per}} +
V_{\bom} ,
\eeq
where: $\Delta$ is the $d$-dimensional Laplacian operator; $V_{\mathrm{per}}$ is a bounded $\Z^d$-periodic potential; and  $V_{\bom}$ is an Anderson-type random potential:
\beq
V_{\bom} (x):= 
\sum_{j \in \Z^d} \omega_j \,  u_j(x), \quad \text{with } \quad u_j(x )=u(x-j) ,\label{AndV}
\eeq
where   
the single site potential $u$ is a  nonnegative bounded 
measurable function
on $\R^{d}$ with compact support, uniformly 
bounded away from zero in
a neighborhood of the origin, and
$\bom=\{ \omega_j \}_{j\in
\Z^d}$ is a family of independent 
identically distributed random
variables,  whose  common probability 
distribution $\mu$ is non-degenerate with a bounded density $\rho$ with compact support.

We  normalize   $H_{\bom}$ as follows.   We  first require  $\inf \supp \mu = 0$, which  can always be realized by 
changing   the periodic potential  $V_{\mathrm{per}}$. Second, we  set $\norm{u}_\infty=1$, which can achieved by rescaling $\mu$.
We then adjust    $V_{\mathrm{per}}$ by adding a constant so $\inf \sigma\pa{ -\Delta + V_{\mathrm{per}}}= 0$,  in which case $[0, E_*] \subset\sigma\pa{ -\Delta + V_{\mathrm{per}}}$  for some $E_*>0$.   Thus, without loss of generality, we will assume that the random Schr\"odinger 
operator $H_\omega$  given in \eq{AndH}-\eq{AndV} is normalized as follows:
 \begin{itemize} 
 
 \item[(I)]  The free Hamiltonian $ H_0 := -\Delta + V_{\mathrm{per}}$ has $0$ as the bottom of its spectrum:
  \beq
\inf \sigma(H_0)= 0 .
\eeq
\item[(II)] The single site potential $u$ is a measurable function  on $\R^d$
such that
 \begin{equation} \label{u}
\norm{u}_\infty=1 \quad \text{and} \quad u_{-}\chi_{\Lambda_{\delta_{-}}(0)}\le u \le \chi_{\Lambda_{\delta_{+}}(0)}\quad \text{with $u_{-}, \delta_{\pm}\in ]0,\infty[ $};
\end{equation}
we set
\beq
U_+:= \norm{\textstyle{\sum}_{j \in \Z^d} \,  u_j}_\infty\le \max\set {1, \delta_+^d}.   \label{U+}
\eeq

\item[(III)]  $\bom=\{ \omega_j\}_{j \in \Z^d}$ is a family of independent, identically distributed random variables, whose common probability distribution $\mu$  has a density $\rho$ such that
\beq \label{mu}
\set{0  ,M_\rho}\in \esssupp \rho \subset   [0,M_\rho],\; \text{with} \;   M_\rho\in ]0,\infty[   \;\text{and} \; \rho_+:=\norm{\rho}_\infty < \infty. 
\eeq
\end{itemize}
A random Schr\"odinger operator $H_\omega$ on  $\mathrm{L}^2(\mathbb{R}^d)$   as in \eq{AndH}-\eq{AndV}, normalized as in (I)-(III), will be called an  \emph{Anderson Hamiltonian}. The common probability distribution $\mu$  in (III) is said to be a uniform-like distribution if its density $\rho$ also satisfies $\rho_-:=\essinf   \rho  \chi_{[0, M_\rho]}>0$, in which case we have
\beq \label{unifdist}
\rho_-  \chi_{[0, M_\rho]}\le  \rho\le  \rho_+  \chi_{[0, M_\rho]} \quad \text{with} \quad 
\rho_\pm,M_\rho \in  ]0,\infty[ .
\eeq

An  Anderson Hamiltonian $H_{\bom}$  is a $\Z^d$-ergodic family of
random self-adjoint operators.
It follows from standard results (cf.\  \cite{KM,CL,PF})
that there exists fixed subsets $\Sigma$,  $\Sigma_{\mathrm{pp}}$, $\Sigma_{\mathrm{ac}}$  and $\Sigma_{\mathrm{sc}}$ of $\R$ so that the spectrum $\sigma(H_{\bom})$
of $H_{\bom}$,  as well as its pure point, 
absolutely continuous, and singular continuous  components,
are equal to these fixed sets with probability one.  With our normalization, 
  the non-random spectrum $\Sigma$ of    an  Anderson Hamitonian  $H_{\bom}$
 satisfies (cf.\  \cite{KiM})
 \beq
 \sigma\pa{ H_0} \subset \Sigma \subset [0, \infty[,
 \eeq
 so $\inf \Sigma=0 $ and $ [0, E_*] \subset\Sigma$ for some $E_*=E_*(V_{\mathrm{per}})>0$.
Note that $ \Sigma= \sigma\pa{ -\Delta } =[0,\infty[  \quad \text{if} \quad  V_{\mathrm{per}}=0.$

 An Anderson Hamiltonian $H_\bom$ exhibits Anderson and dynamical localization at the bottom of the spectrum \cite{HM,CH,Klo93,KSS,GdB,DS,GKboot,GKgafa,AENSS}.  More precisely, there exists an energy $E_1 >0$ such that   $[0, E_1]  \subset\Xi^{\text{CL}}$, where $ \Xi^{\text{CL}}$ is the region of complete localization  for the random operator
$H_\bom$ \cite{GKduke,GKsudec}. (See Appendix~\ref{appMSA} for a discussion of localization. Note that $\R \setminus \Sigma \subset \Xi^{\text{CL}}$ in our definition.)  
Similarly,  given an energy $E_1>0$, we have $[0, E_1]  \subset \Xi^{\text{CL}}$ if $\rho_+$ in \eq{mu} is sufficiently small, corresponding to a large disorder regime.

Finite volume operators will be defined for 
finite boxes $\Lambda=\Lambda_L(j)$, where $j\in \Z^d$ and $L \in 2\N$,  $ L > \delta_+$.   Given such $\Lambda$, we will consider the random Schr\"odinger operator $H_\bom^{(\Lambda)}$ on $\L^2(\Lambda)$ given by the restriction of the Anderson Hamiltonian  $H_\bom$ to $\Lambda$ with periodic boundary condition. To do so, we identify $\Lambda$ with a torus
in the usual way by identifying opposite edges, and  define finite volume operators  
\begin{align}\label{finvolH}
H_{\bom}\up{\Lambda} :=H_{0}\up{\Lambda}+ V_{\bom}\up{\Lambda} \quad \text{on}   \quad \L^{2}(\Lambda).
\end{align}
The finite volume free Hamiltonian $H_{0}\up{\Lambda}$ is given by
\beq
H_{0}\up{\Lambda}:= - \Delta\up{\Lambda} +  V_{\mathrm{per}}\up{\Lambda} \quad \text{on}   \quad \L^{2}(\Lambda),
\eeq
where 
$\Delta\up{\Lambda}$ is the  Laplacian on $\Lambda$ with periodic boundary condition and   
$V_{\mathrm{per}}\up{\Lambda}$ is the restriction of $ V_{\mathrm{per}}$  to $\Lambda$.
The random potential $V_{\bom}\up{\Lambda}$ is the restriction of $V_{\bom\up{\Lambda}}$ to $\Lambda$, where, given $\bom=\set{\omega_i}_{i \in \Z^d}$ , $\bom\up{\Lambda}=\set{\omega\up{\Lambda}_i}_{i \in \Z^d}$ is defined as follows:
\beq\begin{split}
\omega\up{\Lambda}_i &=\omega_i \quad \text{if}  \quad i \in \Lambda, \\
\omega\up{\Lambda}_i &=\omega\up{\Lambda}_k  \quad \text{if}  \quad k-i \in L\Z^d .
\end{split}\eeq
The random finite volume operator $ H_{\bom}\up{\Lambda}$ is covariant with respect to translations in the torus. 
If $B \subset \R$ is a Borel set, we write $P_\bom^{(\Lambda)}(B):=\chi_B\pa{H_\bom^{(\Lambda)}}$ and  $P_\bom(B):=\chi_B(H_\bom)$ for the spectral projections.

The finite volume operator $H_\bom^{(\Lambda)}$ has a compact resolvent, and hence its ($\bom$-dependent) spectrum consists of isolated  eigenvalues with finite multiplicity.  It satisfies a Wegner estimate  \cite{CH,CHK2}:  Given $E_0 > 0$, there exists a  constant $K_W$, independent of $\Lambda$, such that for  all 
intervals $I\subset [0,E_0]$ we have
\beq\label{Wegner}
\E \set {\tr P_\bom^{(\Lambda)}(I)} \le K_W \, \rho_+ \abs{I}\abs{\Lambda}.
\eeq
The constant $K_W$ given in  \cite{CH,CHK2} depends on $E_0,d,  u, V_{\mathrm{per}},M_\rho$, but not on  $\rho_+$.

The integrated density of states (IDS)  for $H_\bom$ is given, for a.e.  $E \in \R$, by 
 \beq \label{N(E)}
 N(E):= \lim_{L \to \infty} \abs{\Lambda_L(0)}^{-1} \tr P_{\bom}^{(\Lambda_L(0))}(]-\infty,E])\quad \text{for $\P$-a.e. $\bom$},
 \eeq
 in the sense that the limit exists and is the same  for $\P$-a.e. $\bom$   (cf.\  \cite{CL,PF}). It follows from \eq{Wegner} that the IDS $N(E)$ is  locally Lipschitz,  hence continuous, so \eq{N(E)} holds for all $E\in \R$.
For all $E\in \R$ we have
 \beq
 N(E) = \lim_{L \to \infty} \E\set{ \abs{\Lambda_L}^{-1} \tr P_{\bom}^{(\Lambda_L)}(]-\infty,E])}.
 \eeq
 $N(E)$ is a nondecreasing absolutely continuous function on $\R$, the cumulative distribution function of the density of states measure, given by
 \beq \label{dsm}
 \eta(B):=     \E \tr \set {\chi^{(1)}_0  P_{\bom}(B)\chi^{(1)}_0}        \quad \text{for a Borel set $B \subset \R$.}
 \eeq
In particular
 $N(E)$ is differentiable a.e.\ with respect to Lebesgue measure,  with  $n(E):=N^\pr(E) \ge 0$ being the density of the measure  $\eta$, so   $n(E)>0$ for  $\eta$-a.e.  $E$.

 Given
 an energy $\cE \in \Sigma$, using \eq{Wegner}
we define a point process $\xi_{\cE,\bom}^{(\Lambda)}$ on the real line   by the rescaled spectrum of the finite volume operator $H_{\bom}^{(\Lambda)}$ near $\cE$:
\beq\label{defxi0}
\xi_{\cE,\bom}^{(\Lambda)}(B):= \tr \set{\chi_B\pa{\abs{\Lambda}\pa{H_{\bom}^{(\Lambda)} -\cE}} }= \tr\set{ P_{\bom}^{(\Lambda)}\pa{\cE + \abs{\Lambda}^{-1} B}}
\eeq
for a Borel set $B \subset \R$. (We refer to \cite{DV} for  definitions and results concerning random measures and point processes.)

\begin{theorem}\label{thmmain} Let $H_\bom$ be an Anderson Hamiltonian with $\delta_- \ge 2$ and a uniform-like distribution $\mu$.
Then there exists an energy $E_0>0 $,   such that:
\begin{itemize}
\item[(a)]  For all energies $\cE \in \Xi^{\text{CL}}\cap  [0, E_0[$ such that the IDS $N(E)$ is differentiable at
$\cE$ with $n(\cE):=N^\pr(\cE) >0$,  the point process $\xi_{\cE,\bom}^{(\Lambda_L)}$ converges weakly, as $L \to \infty$, to the Poisson point process $\xi_\cE$ on $\R$ with intensity measure  $\nu_\cE(B):=\E\,  \xi_\cE(B)= n(\cE)\abs{B}$, i.e., $\di \nu_\cE= n(\cE) \di E$.
\item[(b)] With probability one, every eigenvalue of   $H_\bom$  in $ \Xi^{\text{CL}}\cap  [0, E_0[$ is simple.
\end{itemize}
Similarly,  given an energy  $E_0 >0$, (a) and (b) hold if the probability distribution $\mu$ in \eq{unifdist} has a density $\rho$ with $\tfrac{\rho_+}{\rho_-} \rho_+^{2^d-1}$ sufficiently small.  In fact, there exists a constant $Q_{d,V_{\mathrm{per}} }>0$, such that  (a) and (b) hold whenever
\beq\label{maincond}
U_+{u_-^{-2^d}} \tfrac{\rho_+}{\rho_-} \rho_+^{2^d-1}  \gamma_d(E_0) \min \set{1, E_0^{2^d-d -1}}  \max \set{1 ,  E_0^{2^{d+2}} }\le Q_{d,V_{\mathrm{per}} }, 
 \eeq
where we have  $\gamma_d(E_0)=1$ if $d \ge 2$, and  $\gamma_1(E_0)=\gamma_{1,V_{\mathrm{per}}}(E_0)\in ]0,1]$ with  $\lim_{E_0 \to 0} \gamma_1(E_0)=0$.
\end{theorem}

The next theorem gives our Minami estimate for the continuum Anderson Hamiltonian, a crucial ingredient for proving Theorem~\ref{thmmain}.

 \begin{theorem}\label{thmminami} Let $H_\bom$ be an Anderson Hamiltonian with $\delta_- \ge 2$ and a uniform-like distribution $\mu$.  
 Then there exists a constant $Q_{d,V_{\mathrm{per}} } >0$,  such that whenever \eq{maincond} holds for an energy $E_0>0$,  we have the
 Minami estimate  
\beq \label{Minami}
\E \set{ \pa{\tr P_\bom^{(\Lambda)}(I)}\pa{\tr P_\bom^{(\Lambda)}(I)-1 } }\le K_M \pa{\rho_+ \abs{I}\abs{\Lambda}}^2,
\eeq
 for all  intervals $I \subset [0,E_0]$ and  $\Lambda=\Lambda_L$ with $L \ge L(E_0)$, with a  constant 
 \beq\label{KMest1}
K_M \le  C_{d,V_\mathrm{per},M_\rho} \pa{1 +  E_0}^{4[[\frac d 4]]}.
\eeq
  In more detail:
\begin{itemize}
\item[(i)] If $H_\bom$ is an Anderson Hamiltonian  with $\delta_- \ge 2$, 
there exists a constant $C_{d,V_{\mathrm{per}}}$ such that, given an energy $E_0>0$, the Wegner estimate  \eq{Wegner} holds for all  
intervals $I\subset [0,E_0]$ with a constant
\beq \label{KW}
K_W \le C_{d,V_{\mathrm{per}}} {u_-^{-2^d}} \rho_+^{2^d-1}  \gamma_d(E_0) \min \set{1, E_0^{2^d-d -1}}  \max \set{1 ,  E_0^{2^{d+2}} },
\eeq
where we have $\gamma_d(E_0)=1$ if $d \ge 2$, and  $\gamma_1(E_0)=\gamma_{1,V_{\mathrm{per}}}(E_0)\in ]0,1]$ with  $\lim_{E_0 \to 0} \gamma_1(E_0)=0$.
\item[(ii)]    If $H_\bom$ is an Anderson Hamiltonian  with a uniform-like distribution $\mu$,
and  for a given $E_0>0$ the constant $K_W$ in \eq{Wegner} satisfies  
\beq \label{KWrhocond}
2 K_W {U_+} \tfrac{\rho_+}{\rho_-} \le 1,
\eeq 
then \eq{Minami} holds  for all  
intervals $I\subset [0,E_0]$ with a constant $K_M= C_{d,V_\mathrm{per},u, M_\rho,E_0} K_W$.
If in addition $\delta_- \ge 2$, we have \eq{KMest1}.
\end{itemize}
\end{theorem}

 Our approach to Minami's estimate is discussed in Section~\ref{secnewap}, where it is illustrated by a proof of the estimate for the (discrete) Anderson model (Theorem~\ref{thmM}). We also comment on the differences between the discrete and the continuum cases.

 On the lattice (the Anderson model), the Wegner estimate  \eq{Wegner} is a simple consequence of spectral averaging (cf.\  \eq{wegnerd}), and  holds  with $K_W=1$ for all $E_0$  \cite{W,FS,CKM,Ki}.  On the continuum the Wegner estimate, which  has not been as simple to prove,  comes with an $E_0$ dependent constant $K_W$ (which also  depends on $d$, $V_{\mathrm{per}}$, and $u$)  \cite{CH,CHK2}.
 The proof given in \cite{CH} requires the covering condition $\delta_- \ge 1$. It allows estimates of the constant, but the estimates do not go to $0$ as either  $E_0 $ or $\rho_+$ go to $0$.  The proof in \cite{CHK2}  does not require a covering condition, but it uses   \cite[Proposition~1.3]{CHK1} (cf.\  \cite[Theorem~2.1]{CHK2}), which relies on the unique continuation principle to show that some constant is strictly positive, giving no control on the constant in \eq{Wegner}.  To prove that \eq{KWrhocond} holds, so we have \eq{Minami}, we need suitable control  of the constant $K_W$, as in \eq{KW}. To obtain this control we introduce a double averaging procedure which uses the covering condition $\delta_- \ge 2$.

Note that the estimate \eq{KW} provides a bound on the differentiated density of states $n(E):=N^\pr(E)$  in the interval  $[0,E_0]$, whenever it exists, since  it then  follows from \eq{Wegner} and \eq{KW} that 
\beq
n(E) \le C_{d,V_{\mathrm{per}}} {u_-^{-2^d}} \rho_+^{2^d}  \gamma_d(E) \min \set{1, E^{2^d-d -1}}  \max \set{1 ,  E^{2^{d+2}}}.
\eeq

Once we have the Minami estimate  \eq{Minami}, we may prove Poisson statistics and simplicity of eigenvalues.
The next theorem is proven for arbitrary Anderson Hamiltonians.

\begin{theorem}\label{thmPoisson} Let $H_\bom$ be an Anderson Hamiltonian.    Suppose there exists an open interval
$\mathcal{I} $ such  that for all large boxes $\Lambda$ the estimate \eq{Minami} holds for any
interval $I \subset \mathcal{I}$ with $\abs{I}\le \delta_0$, for some $\delta_0 >0$, with some constant $K_M$.
Then
\begin{itemize}
\item[(a)]  For all energies $\cE \in\mathcal{I}\cap  \Xi^{\text{CL}}$ such that the IDS $N(E)$ is differentiable at
$\cE$ with $n(\cE):=N^\pr(\cE) >0$,  the point process $\xi_{\cE,\bom}^{(\Lambda_L)}$ converges weakly, as $L \to \infty$, to the Poisson point process $\xi_\cE$ on $\R$ with intensity measure  $\nu_\cE(B):=\E\,  \xi_\cE(B)= n(\cE)\abs{B}$, i.e., $\di \nu_\cE= n(\cE) \di E$.
\item[(b)] With probability one, every eigenvalue of   $H_\bom$  in $\mathcal{I}\cap  \Xi^{\text{CL}}$ is simple.
\end{itemize}
\end{theorem}

Theorem~\ref{thmPoisson}(a) is proven by approximating the point process $\xi_{\cE,\bom}^{(\Lambda_L)}$  by  superpositions of independent point processes, as in \cite{Mo,Mi}, which are then shown to converge weakly to the desired Poisson point process. But here our proof diverges from Minami's, who used the connection, valid for the Anderson model,  between the Borel transform of the density of states measure $\eta$ and    averages of the matrix elements of the imaginary part of the resolvent, to identify the intensity measure of the limit  point process.  Instead, we introduce the random measures
\beq\label{deftheta}
\theta_{\cE,\bom}^{(\Lambda)}(B):= \tr \set{\chi_{\Lambda}P_{\bom}(\cE + \abs{\Lambda}^{-1} B) \chi_{\Lambda}} \quad \text{for a Borel set $B \subset \R$},
\eeq
justified by  \eq{Wegner}-\eq{dsm},
which we show to have the same weak limit as the point processes $\xi_{\cE,\bom}^{(\Lambda)}$, and  use them  to show that, thanks to the Lebesgue Differentiation Theorem, the intensity measure  $\nu_\cE$ of the limit point process $\xi_\cE$ satisfies $\di \nu_\cE= n(\cE) \di E$.

Theorem~\ref{thmmain}  follows immediately by combining Theorem~\ref{thmminami} and Theorem~\ref{thmPoisson}.
Theorem~\ref{thmminami} is proven in Sections~\ref{secWegner} and \ref{secMinami}.  In Section~\ref{secWegner} we prove Wegner estimates with control of the constant in Lemma~\ref{thmWsmall}, and a Wegner estimate with one random variable $\omega_j$ fixed in Lemma~\ref{lemWegner0}.  Theorem~\ref{thmminami}(i) follows from  Lemma~\ref{thmWsmall}(i).  Section~\ref{secMinami} contains the proof of Minami's estimate:  Theorem~\ref{thmminami}(ii) is proven in Lemma~\ref{lemMinami}(i), completing the proof of Theorem~\ref{thmminami}.
Theorem~\ref{thmPoisson} is proven in Sections~\ref{sectPoisson} and \ref{secsimple}. 
In Section~\ref{sectPoisson} we prove Theorem~\ref{thmPoisson}(a), namely the convergence of the rescaled eigenvalues  to a Poisson point process. 
Finally, in Section~\ref{secsimple}  we discuss how     Theorem~\ref{thmPoisson}(b) follows from the Minami estimate \eq{Minami} and \cite{KM}.

Some comments about our notation:
  Finite volumes  will always be understood to be
boxes  $\Lambda=\Lambda_L(j_0)$ with $j_0 \in \Z^d$ and $L \in 2 \N$, $L > \delta_+$. We will always identify such $ \Lambda$ with the  torus $j_0 + \R^d\slash L\Z^d$ . 
 If $j \in \widetilde{\Lambda}$, we will consider sub-boxes  $\Lambda\up{\Lambda}_s(j)$  of $\Lambda$, where   $0< s \le L$, defined by
$
\Lambda\up{\Lambda}_s(j):= \set{\bigcup_{k \in L\Z^d} \Lambda_s(j+k)} \cap \Lambda,
$
i.e., $\chi_{\Lambda\up{\Lambda}_s(j)}:= \chi_{\Lambda}  \sum_{k \in L\Z^d} \chi_{ \Lambda_s(j+k)} $.
 Similarly, we define functions $u\up{\Lambda}_j$ on the torus  $\Lambda$ by  $u\up{\Lambda}_j:=\chi_\Lambda \sum_{k \in L\Z^d} u_{j+k}$, i.e.,  the function $u_j$ will be assumed to have been wrapped around the torus $\Lambda$.  Note that we then have
  $V\up{\Lambda}_\bom=\sum_{j \in \widetilde{\Lambda}} \omega_j u\up{\Lambda}_j$.
  We will  abuse the notation    and just write $\Lambda_s(j)$ for  $ \Lambda\up{\Lambda}_s(j)$,  $u_j$ for $u\up{\Lambda}_j$, and $V\up{\Lambda}_\bom=\sum_{j \in \widetilde{\Lambda}} \omega_j u_j$.  In addition,
given $j \in \Upsilon\cap \Z^d$, where $ \Upsilon=\Lambda_L(0)$  or $\R^d$, we
write $\bom =(\bom_j^\perp,\omega_j)$,  and $H\up{\Upsilon}_{(\omega_j^\perp,\omega_j=s)}=H\up{\Upsilon}_{(\omega_j^\perp,s)}$,  $P\up{\Upsilon}_{(\omega_j^\perp,\omega_j=s)}(I)=P\up{\Upsilon}_{(\omega_j^\perp,s)}(I)$
when we want to make explicit that  $\omega_j=s$.


\section{A new approach to Minami's estimate illustrated by a proof for  the (discrete) Anderson Model}\label{secnewap}

The starting point  (and key idea) in our approach is contained in the following simple  lemma.

\begin{lemma}\label{lemkey} Consider the self-adjoint operator  $
H_{s} = H_0 + s W $ on  the Hilbert space $ \H$,
where $H_0$ and $W$ are  self-adjoint operators on  $\H$, with $W\ge 0$  bounded,
and $s\ge 0$.  Let $P_s(J)= \chi_J(H_s)$ for an interval $J$, and suppose
$\tr P_s(]-\infty,c])<\infty$ for all $c \in \R$ and $s\ge 0$.  Then, for all $a,b \in \R$ with $a <b$
we have
\beq
\tr P_s(]a,b])\le \set{ \tr P_{0} (]-\infty,b]) - \tr P_{t}(]-\infty,b])}
 + \tr P_{t}(]a,b]) \quad \text{for}\quad 0\le s \le t. \label{decomp}
\eeq
\end{lemma}

\begin{proof} Let $a,b \in \R$ with $a <b$ and 
 $0\le s \le t$.   Then, since $W\ge 0$, 
\begin{align}\notag
\tr P_{s}(]a,b]) & = \tr P_{s}(]-\infty,b]) - \tr P_{s}(]-\infty,a]) \\
 &\le  \tr P_{0} (]-\infty,b]) - \tr P_{t}(]-\infty,a])  \label{xi}
\\
& = \tr P_{0} (]-\infty,b]) - \tr P_{t}(]-\infty,b])
 + \tr P_{t}(]a,b]). \notag  \qedhere
\end{align}
\end{proof}

 We will also use the basic spectral averaging estimate:
Let $H_0$ and $W$
 be   self-adjoint operators on a Hilbert space $\H$, with $W\ge 0$ bounded.   Consider the random operator $H_\xi:= H_0  + \xi W$, where  $\xi$ is a random
variable with a non-degenerate probability distribution $\mu$ with compact support.  The basic spectral averaging estimate for such  perturbations of  self-adjoint operators says that, given $\vphi\in \H$ with
$\norm{\vphi}=1$, then for  all bounded
intervals $I\subset \R$ we have (\cite[Corollary~4.2]{CH}, \cite[Eq.~(3.16)]{CHK2})
\begin{equation}\label{sa}
\E_{\xi}\set{ \langle \vphi,\sqrt{W} \chi_I(H_\xi)\sqrt{W} \vphi\rangle}:= \int \di \mu(\xi )\,
 \langle \vphi,\sqrt{W} \chi_I(H_\xi) \sqrt{W}\vphi\rangle \le Q_\mu(\abs{I}),
\end{equation}
where
\beq
Q_\mu(s):=
\left\{ 
\begin{array}{ll}
\rho_\infty s &\mbox{if $\mu$ has a bounded density $\rho$ as in \eq{mu}}
\\
 8 \sup_{a \in \R}
\mu([a,a+s]) &\mbox{otherwise}
\end{array}.
\right.\notag
\eeq
As a consequence, given a trace class operator $S\ge 0$ on $\H$, we have
\begin{equation}\label{sa1}
\E_{\xi}\set{ \tr \set{\sqrt{W}\chi_I(H_\xi) \sqrt{W}S}} \le \pa{\tr S} Q_\mu(\abs{I}).
\end{equation}
Note that  the measure $\mu$ has no atoms if and only if  $\lim_{s
\downarrow 0} Q_\mu(s)=0$.

Lemma~\ref{lemkey} will allow the decoupling of random variables for the  performance of two spectral averagings.

We will first illustrate our approach to Minami's estimate  by giving a simple and transparent proof of the estimate for  in the discrete case, i.e., for the Anderson model.   We will then  comment on how to proceed in the continuum case, i.e., for the Anderson Hamiltonian.

\subsection{Minami's estimate for the (discrete) Anderson model}  An Anderson model will be a discrete random Schr\"odinger operator of the form
\begin{equation}
H_{\bom} = H_0 + V_{\bom} \quad \text{on}\quad \ell^2(\Z^d),\label{defHdisc}
\end{equation}
where $H_0$ is a bounded self-adjoint operator and $V_{\bom}$ is the random potential
given by $V_{\bom}(j)= \omega_j$ for $j \in \Z^d$, where $\bom=\{ \omega_j \}_{ j\in
\Z^d}$ is a family of independent, identically distributed
random
variables  with common
probability distribution $\mu $.     (The usual Anderson model has $H_0=-\Delta$, where $\Delta$  is the discrete Laplacian.)  
 We assume $\mu$ has  compact support and no atoms. Adjusting $H_0$ and $\mu$, we may assume
\beq \label{mudisc}
\set{0  ,M}\in \supp \mu \subset   [0,M]\quad \text{with}  \quad M\in ]0,\infty[ . 
\eeq

Restrictions of $H_\bom$ to finite  volumes $\Lambda\subset\Z^d$  are denoted by
$H_{\bom}\up{\Lambda}$, a self-adjoint operator of the form
\begin{equation}
H_{\bom}\up{\Lambda} = H_{0}\up{\Lambda} + V_{\bom}\up{\Lambda}\quad
\text{on}\quad \ell^2(\Lambda),\label{defHfinitedisc}
\end{equation}
where $H_{0}\up{\Lambda}$ is  a self-adjoint restriction of $H_0$ to the finite-dimensional
Hilbert space  $ \ell^2(\Lambda)$, and $V_{\bom}\up{\Lambda}$ is the restriction of $V_\bom$ to $\Lambda$.
(In the discrete case our results are not sensitive to the choice of $H_{0,\Lambda}$, they hold for any
boundary condition.)   Given a Borel set $J \subset \R$, we write
$P^{(\Lambda)}_{\bom}(J)=P^{(\Lambda)}_{H_\bom}(J)=\chi_J(H_{\bom}\up{\Lambda})$ for  the
associated spectral projection.  

What makes the discrete case much easier than the continuum is that in the discrete case finite volume operators are finite-dimensional and  each random variable couples a rank one perturbation.  
Given a unit vector $\vphi$ in a Hilbert space $\H$, we let
$\Pi_\vphi $ denote the  orthogonal projection onto $\C \vphi$,  the one-dimensional
subspace spanned by $\vphi$.  With this notation, the potentials in 
in \eq{defHdisc} and \eq{defHfinitedisc} are given by sums of rank one perturbations:
\beq
 V_{\bom}= \sum_{j\in \Z^d}\omega_j  \Pi_{j} \quad \text{and}\quad V_{\bom}\up{\Lambda}= \sum_{j\in \Lambda}\omega_j  \Pi_{j}, \quad \text{with}\quad  \Pi_{j}=\Pi_{\delta_j}.
\eeq

For rank one perturbations Lemma~\ref{lemkey} has the following  consequence.

\begin{lemma}\label{lemkey2}  Let $H_s$ be as in Lemma~\ref{lemkey} with
$W=\Pi_\vphi $ for some unit vector $\vphi \in \H$.  Then, for all $a,b \in \R$ with $a <b$  we have
\beq
\tr P_s(]a,b])\le 1 + \tr P_t(]a,b]) \quad \text{for all} \quad 0\le s \le t. \label{decompd2}
\eeq
\end{lemma}

\begin{proof} Let  $0\le s \le t$. Recall that for any  $c \in \R$   we always have
\beq \label{rankone}
0\le \tr P_{s}(]-\infty,c])-\tr P_{t}(]-\infty,c])\le 1,
\eeq
the last inequality being a consequence of the min-max principle applied to rank one
perturbations,  e.g.  \cite[Lemma~5.22]{Ki}.  Thus \eq{decompd2} follows immediately from \eq{decomp}.
\end{proof}

For rank one perturbations  the fundamental spectral averaging estimate \eq{sa} may be stated as follows: 
Consider the random self-adjoint operator
\begin{equation}
H_{\xi} = H_0 + \xi \Pi_\vphi \quad \text{on}\quad \H,\label{defH1}
\end{equation}
where $H_0$ is a  self-adjoint operator on the Hilbert space $\H$, $\vphi\in \H$ with
$\norm{\vphi}=1$, and $\xi$ is a random
variable with a non-degenerate probability distribution $\mu$ with compact support.   Let $P_\xi(J)=\chi_J(H_\xi) $ for a Borel set $J
\subset \R$. Then  for  all bounded
intervals $I\subset \R$ we have \cite{W,FS,CKM,Ki,CH,CHK2}
\begin{equation}\label{sad}
\E_{\xi}\set{ \langle \vphi, P_\xi (I) \vphi\rangle}:= \int \d \mu(\xi)\,
\langle \vphi, P_\xi (I) \vphi\rangle \le Q_\mu\pa{\abs{I}}.
\end{equation}

The Wegner estimate  for an Anderson model \cite{W,FS,CKM,Ki}
 is an immediate consequence of \eq{sad}:
\begin{equation}\label{wegnerd}
\E  \set{\tr P^{(\Lambda)}_{H_\bom}(I)}= \sum_{j \in \Lambda} \E_{\bom^\perp_j}
\set{\E_{\omega_j} \set{ \langle \delta_j,  P^{(\Lambda)}_{H_\bom}(I) \delta_j\rangle}}
\le Q_{\mu}\pa{\abs{I}} |\Lambda|.
\end{equation}

We can now prove Minami's estimate for an Anderson model for arbitrary $\mu$ with compact support and no atoms, a result previously known only  for $\mu$ with a bounded density \cite{Mi,BHS,GV}.  

\begin{theorem}\label{thmM} Let $H_\omega$ be an Anderson model as in \eq{defHdisc}, with $\mu$ arbitrary except for  compact support and no atoms.  Let  $\Lambda\subset\Z^d$ be a finite volume.  For any bounded interval $I$ we have
\beq
\E \left\{ \pa{\tr  P^{(\Lambda)}_{\bom}(I)} \pa{\tr  P^{(\Lambda)}_{\bom}(I)-1} \right\}
  \le
\pa{ Q_\mu\pa{\abs{I}} \abs{\Lambda}}^2 .\label{minamid}
\eeq
\end{theorem}

Theorem~\ref{thmM} is extended in \cite{CGK}, allowing for  $n$ arbitrary intervals and arbitrary single-site probability measure $\mu$ with no atoms. We also give
applications of \eq{minamid}, deriving  new results about the multiplicity of eigenvalues and Mott's
formula for the ac-conductivity when the single site probability distribution is H\"older
continuous. 

\begin{proof}[Proof of Theorem~\ref{thmM}]
Fix $\Lambda\subset \Z^d$ and let $I$
be a bounded interval.  Since the measure $\mu$ has no atoms, it follows from  \eq{wegnerd}
that $\E_\bom\set{ \tr P_\bom^{(\Lambda)}(\{c\})}=0$ for any $c \in \R$. Thus we may
take all intervals to be of the form $]a,b]$, and use  Lemma~\ref{lemkey2}  to decouple the random variable $\omega_j$ from the  random variables $\bom_j^\perp$.
In view of  \eq{mudisc}, for all $\tau_j \ge M$, $j \in \Z^d$,  we  have
\begin{align}\label{choicetau}
\pa{\tr P^{(\Lambda)}_\bom(I)} \pa{\tr P^{(\Lambda)}_\bom(I)-1}
& =
\sum_{j\in\Lambda} \set{ \scal{\delta_j, P^{(\Lambda)}_\bom(I) \delta_j}\pa{\tr P^{(\Lambda)}_\bom(I)-1}  } \\
& \le
\sum_{j\in\Lambda} \set{ \scal{\delta_j, P^{(\Lambda)}_{(\bom_j^\perp,\omega_j)}(I) \delta_j}
\pa{\tr P^{(\Lambda)}_{(\bom_j^\perp,\tau_j)}(I)}}. \notag
\end{align}
We now  average over the
random variables $\bom=\{ \omega_j \}_{ j\in
\Z^d}$.   Using
\eqref{sad}, we get
\begin{align}\label{minamiproofd}
&\E_\bom \set{\pa{\tr P^{(\Lambda)}_\bom(I)}\pa{\tr P^{(\Lambda)}_\bom(I)-1}}\\
& \qquad \qquad \le
\sum_{j\in\Lambda} \E_{\bom_j^\perp} \set{  \pa{\tr P^{(\Lambda)}_{(\bom_j^\perp,\tau_j)}(I)
}\pa{\E_{\omega_j} \set{\scal{\delta_j, P^{(\Lambda)}_{(\bom_j^\perp,\omega_j)}(I) \delta_j} } }
}\notag \\
& \qquad \qquad \le
Q_\mu\pa{\abs{I}} \, \sum_{j\in\Lambda} \E_{\bom_j^\perp} \set{\tr
P^{(\Lambda)}_{(\bom_j^\perp,\tau_j)}(I)}.\notag
\end{align}
This holds for all  $\tau_j \ge M$, $j\in \Z^d$, so we now take
$\tau_j= M + \tilde{\omega}_j$, where $\tilde{\bom}=\set{
\tilde{\omega}_j}_{j\in\Z^d}$ and  $\bom=\set{{\omega}_j}_{j\in\Z^d}$ are two independent,
identically distributed collections of random variables.  Now ${\btau}=\set{{\tau}_j}_{j\in\Z^d}$ are independent identically distributed random variables with a common probability distribution $\mu_{\btau}$ such that
$Q_{\mu_{\btau}}=Q_\mu$.
 We get
\begin{align}\notag
&\E_\bom \set{\pa{\tr P^{(\Lambda)}_\bom(I)}\pa{\tr P^{(\Lambda)}_\bom(I)-1} }=\E_{\btau}\set{  \E_\bom
\set{\pa{\tr P^{(\Lambda)}_\bom(I)}\pa{\tr P^{(\Lambda)}_\bom(I)-1} }}\\
&\qquad \qquad  \qquad \le Q_\mu\pa{\abs{I}} \sum_{j\in\Lambda}
\E_{(\bom_j^\perp,\tau_j)} \pa{\tr P^{(\Lambda)}_{(\bom_j^\perp,\tau_j)}(I)} \le
(Q_\mu\pa{\abs{I}} |\Lambda|)^2, \label{pre-minami}
\end{align}
where we used  the  Wegner estimate \eqref{wegnerd}. (More precisely, we estimate as in \eqref{wegnerd}; the random variables do not need to be identically distributed.)
\end{proof}

\subsection{Stepping up to the continuum}\label{subsectcomments} 
 Unfortunately things are not so simple for the continuum Anderson Hamiltonian.  
    The main reason is that the random potential  $V_\bom$ in \eq{AndV}  is a sum of independent random perturbations of infinite rank, not rank one as in the discrete case, and thus  the a priori bound  in  \eq{rankone}, and  also
 Lemma~\ref{lemkey2}, are not applicable anymore.  

To prove Minami's estimate on the continuum we will use the fundamental spectral averaging estimate as in  \eq{sa1}.
The straightforward expansion of the trace in \eq{wegnerd} and \eq{minamiproofd}  cannot be used for the spectral averaging, even with $u_j$ instead of   $\delta_j$, and   will be replaced by a more sophisticated expansion in terms of trace class operators, as in \cite{CH,CHK2} (cf.\  \eq{trPT}-\eq{uniftracesum2}). Lemma~\ref{lemkey} will be modified, since the term in brackets in  \eq{decomp} does not satisfy an a priori bound as in \eq{rankone} anymore.    This term will be estimated using the Birman-Solomyak formula (cf.\   \eq{decomp2},\eq{xib}).  The bound in \eq{rankone} is then replaced by averaging the resulting expression over all the other random variables and using the Wegner estimate \eq{Wegner} (cf.\  \eq{avessf}).  The resulting  bound is useful if the constant $K_W$ in \eq{Wegner} is not too big (we have $K_W=1$ in the lattice, as can be seen in \eq{wegnerd}).  Since previous  proofs of the Wegner estimate do not give the desired control of $K_W$,  we must revisit  the Wegner estimate. We introduce a double averaging procedure that provides the desired estimates on the constant $K_W$
(cf.\  Lemma~\ref{thmWsmall}).  In addition, because of the way we use the Birman-Solomyak formula, we do not have freedom in the choice of $\tau_j$ as in \eq{choicetau}, we have to take $\tau_j=M_\rho$.  Thus we cannot average in $\btau$ as in  \eq{pre-minami}; this argument is replaced by a  refinement of the Wegner estimate
where one of the random variables is fixed (cf.\  Lemma~\ref{lemWegner0}).


\section{The Wegner estimate revisited}\label{secWegner}

Let $H_\bom$ be the Anderson Hamiltonian, $E_0 >0$,  $I \subset [0,E_0]$ an interval, and $\Lambda$ a finite box.    To prove the Wegner estimate \eq{Wegner}, it is shown in \cite{CH,CHK2} that 
\beq\label{trPT}
\tr P_\bom^{(\Lambda)}(I) \le Q_1 \sum_{j,k \in \widetilde{\Lambda}} \abs{
\tr \set{\sqrt{u_k}  P_\bom^{(\Lambda)}(I)\sqrt{u_j } \, T_{j,k}^{(\Lambda)}  } },
\eeq
where $\set{T_{j,k}^{(\Lambda)} }_{j,k \in \widetilde{\Lambda}}$ are (non-random) trace class operators in $\L^2(\Lambda)$  such that
\beq\label{uniftracesum}
\max_{j \in  \widetilde{\Lambda}} \set{\sum_{k\in  \widetilde{\Lambda}} \norm{T_{j,k}^{(\Lambda)}}_1} \le Q_2,
\eeq
with  $Q_1,Q_2$  constants depending only on 
 $E_0,d,u,V_{\mathrm{per}},M_{\rho}$.  Letting  $T_{j,k}^{(\Lambda)}= U_{j,k}^{(\Lambda)}\abs{T_{j,k}^{(\Lambda)}}$ be the polar decomposition of the operator   $T_{j,k}^{(\Lambda)}$,  recalling that then     $\abs{T_{j,k}^{(\Lambda)*}}=U_{j,k}^{(\Lambda)}T_{j,k}^{(\Lambda)}U_{j,k}^{(\Lambda)*} $, and setting
 \beq
 S_j^{(\Lambda)}:= \tfrac 12 \sum_{k \in \widetilde{\Lambda}} \pa{\abs{T_{j,k}^{(\Lambda)*} } + \abs{T_{k,j}^{(\Lambda)} }} \ge 0 \quad \text{for} \quad j \in \widetilde{\Lambda} ,
 \eeq
    we obtain
 \begin{align} \label{traceclassexp}
\tr P_\bom^{(\Lambda)}(I)
 \le   {Q_1}  \sum_{j \in \widetilde{\Lambda}} 
\tr \set{\sqrt{u_j}  P_\bom^{(\Lambda)}(I)\sqrt{u_j } \,S_j^{(\Lambda)} }  ,
\end{align}
with
\beq\label{uniftracesum2}
\max_{j\in  \widetilde{\Lambda}} \set{\tr S_j^{(\Lambda)} }\le Q_2.
\eeq
If we now take the expectation in \eq{traceclassexp}, use \eq{sa1} and \eq{uniftracesum2}, we get the Wegner estimate \eq{Wegner} with  $K_W= Q_1Q_2$.

We will need control of the constant $K_W$  and  a Wegner estimate with one of the random variables, say  $\omega_0$, fixed. In the course of obtaining control over $K_W$ 
we will derive \eq{trPT} with estimates on the constants $Q_1$ and $Q_2$ in the case when $\delta_- \ge 1$.

\subsection{ A Wegner estimate with control of the constants}

\begin{lemma}\label{thmWsmall}  Let $H_\bom$ be an Anderson Hamiltonian.   
\begin{itemize}
\item[(i)] Assume  $\delta_- \ge 2$.  Then   there exists a constant $C_{d,V_{\mathrm{per}}}$ such that, given an energy $E_0>0$, \eq{Wegner} holds for all  
intervals $I\subset [0,E_0]$ with a constant
\beq
K_W \le  C_{d,V_{\mathrm{per}}} \pa{ \tfrac {\rho_+}{u_-}}^{2^d} \gamma_d(E_0) \min \set{1, E_0^{2^d-d -1}}  \max \set{1 ,  E_0^{2^{d+2}} },
\eeq
where we have  $\gamma_d(E_0)=1$ if $d \ge 2$, and  $\gamma_1(E_0)=\gamma_{1,V_{\mathrm{per}}}(E_0)\in ]0,1]$ with  $\lim_{E_0 \to 0} \gamma_1(E_0)=0$.

\item[(ii)]  Assume  $\delta_- \ge 1$.  Then,  given an energy $E_0>0$,  \eq{trPT}-\eq{uniftracesum2} hold for all 
intervals $I\subset [0,E_0]$ with  constants
\beq  \label{Q12}
Q_1 =   \pa{1 +  E_0}^{2[[\frac d 4]]} \quad\text{and} \quad Q_2 =C_{d,V_{\mathrm{per}}}^\pr, 
\eeq
and hence  \eq{Wegner} holds for   all 
intervals $I\subset [0,E_0]$ with a constant
\beq \label{KW12}
K_W \le C_{d,V_{\mathrm{per}}}^\pr  \pa{1 +  E_0}^{2[[\frac d 4]]}.
\eeq
\end{itemize}

\end{lemma}

\begin{proof}   Assume $\delta_- \ge m$, where $m$ is either $1$ or $2$.
We set  $\chi\up{m}_j=  \chi_{\Lambda_m(j)}$ for $j \in \widetilde{\Upsilon}:=\Upsilon\cap \Z^d$, where $\Upsilon$ is either $\R^d$ or a finite box $\Lambda$ (recall that in this case $ \chi_{\Lambda_m(j)}$ denotes  $ \chi\up{\Lambda}_{\Lambda_m(j)}$, a sub-box in the torus).
Note that for any $j_0 \in \widetilde{\Upsilon}$ we have 
\begin{equation}\label{cover211}
 \sum_{j\in \pa{j_0 + m \Z^d}\cap \Upsilon} \chi\up{m}_j = 1.
\end{equation}
We also let  $\hat{\chi}\up{m}_j = {u_j}^{-\frac 12}\chi\up{m}_j$ on  $\Lambda_m(j)$,  ${\hat{\chi}\up{m}}_j =0$ otherwise. It follows from \eq{u} that
${\hat{\chi}\up{m}}_j \le u_-^{-\frac 12} \chi\up{m}_j$.  (Recall we write  $u_j$ for $u\up{\Lambda}_j$.)

To prove (i), assume  $\delta_- \ge 2$.
We write $\bom^\pr=\{ \omega_j \}_{j\in
2 \Z^d}$, $\bom^{\pr\pr}=\{ \omega_j \}_{j\notin
2\Z^d}$.  We set 
\begin{equation}
H_{\bom^{\pr\pr}} := H_0 + V_{\bom^{\pr\pr}}, \quad  V_{\bom^{\pr\pr}}:= \sum_{j\notin 2\Z^d} \omega_j u_j.
\end{equation}
Note that  $H_{\bom^{\pr\pr}}$ is a $2\Z^d$ ergodic family of random self-adjoint operators, and we have
\beq  \label{H>VH}
H_{\bom} \ge H_{\bom^{\pr\pr}}  \ge H_0 , \quad   H_{\bom^{\pr\pr}}\ge V_{\bom^{\pr\pr}}.
\eeq

Fix an energy $E_0>0$, a box $\Lambda$, and 
let $I=]a,b]\subset [0,E_0]  $.
Set $p =2^{d+1}$.  Given $t> 0$, the function  $g_t(x)=\pa{1 + tx}^{-2 p}$ is convex on the interval $]- \frac 1 t,\infty[$.  Thus, using \eq{H>VH}, we can proceed as in \cite{CH} using convexity and Jensen's inequality, cf.\  Lemma~\ref{lemconvexJ}, and then \eq{cover211} and \eq{u},   to get 
\begin{align}\notag 
\tr P\up{\Lambda}_{\bom}(I)&  \le \pa{1 + t E_0}^{2 p}  \tr \set{ P\up{\Lambda}_{\bom}(I) \pa{1 + tH\up{\Lambda}_{\bom} }^{-2p} P\up{\Lambda}_{\bom}(I) } \\\notag
& \le\pa{1 + t E_0}^{2 p} \tr\set{ P\up{\Lambda}_{\bom}(I) \pa{1 + tH\up{\Lambda}_{\bom^{\pr\pr}} }^{-2p}  P\up{\Lambda}_{\bom}(I)}\\
&=\pa{1 + t E_0}^{2 p} \tr\set{ P\up{\Lambda}_{\bom}(I) \pa{1 + tH\up{\Lambda}_{\bom^{\pr\pr}} }^{-2p}  }  \label{trPT1}\\
& =   \pa{1 + t E_0}^{2p}\sum_{j,k\in  \Lambda\cap 2\Z^d}   \tr  \set{P\up{\Lambda}_{\bom}(I) \chi\up{2}_j  \pa{1 + tH\up{\Lambda}_{\bom^{\pr\pr}} }^{-2p} \chi\up{2}_k} \notag \\
& =    \pa{1 + t E_0}^{2p}\sum_{j,k\in  \Lambda\cap 2\Z^d}  \tr  \set{\sqrt{u_k}P\up{\Lambda}_{\bom}(I)\sqrt{u_j}{\hat{\chi}\up{2}}_j   \pa{1 + tH\up{\Lambda}_{\bom^{\pr\pr}} }^{-2p} {\hat{\chi}\up{2}}_k} \notag .
\end{align}
It then follows from \eq{sa1}, proceeding  as in \eq{trPT}-\eq {traceclassexp}
(see also  \cite[Lemma~2.1]{CHK2}),  that
\begin{align} \label{suminest}
\E_{\bom^\pr} \tr P\up{\Lambda}_{\bom}(I)& \le    \pa{1 + t E_0}^{2p} \rho_+ \abs{I} \sum_{j,k\in  \Lambda\cap 2\Z^d}   \norm{\hat{\chi}\up{2}_j  \pa{1 + tH\up{\Lambda}_{\bom^{\pr\pr}} }^{-2p} {\hat{\chi}\up{2}}_k}_1\\
&\le   \pa{1 + t E_0}^{2p}u_-^{-1}  \rho_+ \abs{I} \sum_{j,k\in  \Lambda\cap 2\Z^d}   \norm{\chi\up{2}_j  \pa{1 + tH\up{\Lambda}_{{\bom}^{\pr\pr}} }^{-2p} \chi\up{2}_k}_1 .\notag
\end{align}

We now use several deterministic estimates.  First,
\begin{align}\label{det1}
& \norm{\chi\up{2}_j  \pa{1 + t{H\up{\Lambda}_{\bom^{\pr\pr}}} }^{-2p} \chi\up{2}_k}_1\\
& \qquad \qquad  \le \sum_{r \in \Lambda\cap 2\Z^d}  \norm{\chi\up{2}_j  \pa{1 + t{H\up{\Lambda}_{\bom^{\pr\pr}}} }^{-p} \chi\up{2}_r}_2 \norm{\chi\up{2}_r \pa{1 + t{H\up{\Lambda}_{\bom^{\pr\pr}}} }^{-p} \chi\up{2}_k}_2. \notag
\end{align}
Second, 
\begin{align}\label{det2}
& \norm{\chi\up{2}_j  \pa{1 + t{H\up{\Lambda}_{\bom^{\pr\pr}}} }^{-p} \chi\up{2}_r}_2^2 \\
&\qquad \qquad   \qquad   \le 
\norm{\chi\up{2}_j  \pa{1 + t{H\up{\Lambda}_{\bom^{\pr\pr}}} }^{-p} \chi\up{2}_r} \norm{\chi\up{2}_j  \pa{1 + t{H\up{\Lambda}_{\bom^{\pr\pr}}} }^{-p} \chi\up{2}_r}_1 .\notag
\end{align}
Third, we estimate $\norm{\chi\up{2}_j  \pa{1 + tH\up{\Lambda}_{\bom^{\pr\pr}}}^{-p} \chi\up{2}_r} $ using the Combes-Thomas estimate.  We use the precise estimate   provided in
\cite[Eq. (19) in Theorem 1]{GKdecay} (with $\gamma= \frac 1 2$), modified for finite volume operators with periodic boundary condition as in  \cite[Lemma 18]{FKac} and  \cite[Theorem 3.6]{KKdef}, plus the fact that we are using boxes of side $2$.    We have ($L\ge L_d$),  with $d_\Lambda(j,r)$ the distance on the torus $\Lambda$,
\begin{align}\label{CT}
&\norm{\chi\up{2}_j  \pa{1 + t{H\up{\Lambda}_{\bom^{\pr\pr}}} }^{-p} \chi\up{2}_r} = t^{-p}\norm{\chi\up{2}_j  \pa{t^{-1} + H\up{\Lambda}_{{\bom^{\pr\pr}}} }^{-p} \chi\up{2}_r}\\
& \qquad \qquad \qquad \le t^{-p}\pa{\tfrac {4 t} {3}}^{p} \e^{\frac1  {2 \sqrt{t}} }\e^{-\frac {1}{8 \sqrt{td}}d_\Lambda(j,r)}=\pa{\tfrac {4 } {3}}^{p}\e^{\frac1  {2 \sqrt{t}} }\e^{-\frac {1}{8 \sqrt{td}}d_\Lambda(j,r)}.\notag
\end{align}
Fourth, note that
\begin{align}\notag
& \norm{\chi\up{2}_j  \pa{1 + t{H\up{\Lambda}_{\bom^{\pr\pr}}} }^{-p} \chi\up{2}_r}_1\\
&\qquad \qquad \quad \le  \norm{\chi\up{2}_j  \pa{1 + t{H\up{\Lambda}_{\bom^{\pr\pr}}} }^{-\frac p 2}}_2 \norm{\chi\up{2}_r  \pa{1 + t{H\up{\Lambda}_{\bom^{\pr\pr}}} }^{- \frac p 2}}_2
\label{trace1}\\
 & \qquad \qquad \quad  = \norm{\chi\up{2}_j  \pa{1 + t{H\up{\Lambda}_{\bom^{\pr\pr}}} }^{-p} \chi\up{2}_j}_1^{\frac 12} \norm{\chi\up{2}_r  \pa{1 + t{H\up{\Lambda}_{\bom^{\pr\pr}}} }^{-p} \chi\up{2}_r}_1^{\frac 12} .\notag 
 \end{align}
 
We now average   over $\bom^{\pr\pr}$. Using \eq{det1}-\eq{trace1}, we have 
\begin{align}\label{averageeven}
&\E_{\bom^{\pr\pr}} \set{\norm{\chi\up{2}_j  \pa{1 + tH\up{\Lambda}_{\bom^{\pr\pr}} }^{-p} \chi\up{2}_r}_1^{\frac 12}\norm{\chi\up{2}_r \pa{1 + tH\up{\Lambda}_{\bom^{\pr\pr}} }^{-p} \chi\up{2}_k}_1^{\frac 12}}\\ \notag
&\quad  \le \E_{\bom^{\pr\pr}} \left\{{\norm{\chi\up{2}_j  \pa{1 + tH\up{\Lambda}_{\bom^{\pr\pr}} }^{-p} \chi\up{2}_j}_1^{\frac 14} \norm{\chi\up{2}_r  \pa{1 + tH\up{\Lambda}_{\bom^{\pr\pr}} }^{-p} \chi\up{2}_r}_1^{\frac 12} }\right. \\
& \qquad \qquad \qquad  \qquad \qquad \qquad   \qquad \qquad \left.{\times \norm{\chi\up{2}_k \pa{1 + tH\up{\Lambda}_{\bom^{\pr\pr}} }^{-p} \chi\up{2}_k}_1^{\frac 14}}\right\} \notag \\
& \quad \le \beta_t:= \E_{\bom^{\pr\pr}} \set{\norm{\chi\up{2}_{0} \pa{1 + tH\up{\Lambda}_{\bom^{\pr\pr}} }^{-p} \chi\up{2}_{0}}_1},\notag
\end{align}
where we used H\"older's inequality plus translation invariance (in the torus) of the expectation.
  
 It now follows from  from \eq{det1}, \eq{det2},  \eq{CT}, \eq{trace1},  and \eq{averageeven} that
 \begin{align}\notag
&\E_{\bom^{\pr\pr}} \set{ \sum_{j,k\in  \Lambda\cap 2\Z^d}   \norm{\chi\up{2}_j  \pa{1 + tH\up{\Lambda}_{{\bom}^{\pr\pr}} }^{-2p} \chi\up{2}_k}_1 }  \\ \notag
&\quad \quad  \le \beta_t \, \e^{\frac {1}{ 2\sqrt{t}} }  \pa{\tfrac 4 3}^{p} \sum_{j,k, r\in  \Lambda\cap 2\Z^d}  \e^{-\frac {1}{16 \sqrt{td}}d_\Lambda(j,r)}\e^{-\frac {1}{16\sqrt{td}}d_\Lambda(r,k)}\\  \label{longest}
&\quad \quad \le 2^{-d}  \beta_t \, \e^{\frac {1}{ 2\sqrt{t}} }  \pa{\tfrac 4 3}^{p} \abs{\Lambda} \pa{  \sum_{r \in  2\Z^d}\e^{-\frac {1}{16 \sqrt{td}} \abs{r}}}^2 \\ \notag
&\quad \quad=2^{-d}  \beta_t \, \e^{\frac {1}{ 2\sqrt{t}} }  \pa{\tfrac 4 3}^{p} \abs{\Lambda} 
\pa{ \sum_{s\in  \Z}\e^{-\frac {1}{8d \sqrt{t}} \abs{s}}  }^{2d}  \\ \notag
&\quad \quad\le 2^{-d}  \beta_t \, \e^{\frac {1}{ 2\sqrt{t}} }  \pa{\tfrac 4 3}^{p} \abs{\Lambda}\pa{ 1 + 2 \int_0^\infty  \di s  \, \e^{-\frac {1}{8d \sqrt{t}} \abs{s}} }^{2d}\\ \notag
&\quad \quad\le2^{-d}  \beta_t \, \e^{\frac {1}{ 2\sqrt{t}} }  \pa{\tfrac 4 3}^{p} \abs{\Lambda}
\pa{1 + {16d\sqrt{t}}  }^{2d},
 \end{align}
 so we conclude from  \eq{suminest} that
 \begin{align}\label{estimatetrP}
\E_{\bom} \tr P\up{\Lambda}_{\bom}(I)& \le  \pa{\tfrac 4 3}^{p} \tfrac 1 {2{u_-}}  \pa{1 + t E_0}^{2p} \beta_t \, \e^{\frac {1}{ 2\sqrt{t}} }  \pa{1 + 16d\sqrt{t}  }^{2d}   \rho_+ \abs{I}  \abs{\Lambda}.
\end{align}

We now estimate $ \beta_t$.   We  have, using periodicity, and again Lemma~\ref{lemconvexJ} with  $H\up{\Lambda}_{\bom^{\pr\pr}}\ge V_{\bom^{\pr\pr}}$ and \eq{u},
\begin{align}\notag
\beta_t  & := \E_{\bom^{\pr\pr}} \set{\tr \set{\chi\up{2}_{0} \pa{1 + tH\up{\Lambda}_{\bom^{\pr\pr}} }^{-p} \chi\up{2}_{0}}}
= \tfrac {2^d}{\abs{\Lambda} }\E_{\bom^{\pr\pr}} \set{\tr\set{ \pa{1 + tH\up{\Lambda}_{\bom^{\pr\pr}} }^{-p}} } \\ \label{term10}
& \le   \tfrac {2^d}{\abs{\Lambda} }\E_{\bom^{\pr\pr}} \set{\tr \set{\pa{1 + tH\up{\Lambda}_{\bom^{\pr\pr}} }^{-\frac p 4}  \pa{1 + tV_{\bom^{\pr\pr}} }^{-\frac p 2} \pa{1 + tH\up{\Lambda}_{\bom^{\pr\pr}} }^{-\frac p 4}}}\\ 
& =   \tfrac {2^d}{\abs{\Lambda} }\E_{\bom^{\pr\pr}} \set{\tr\set{ \pa{1 + tV_{\bom^{\pr\pr}} }^{-\frac p 4}  \pa{1 + tH\up{\Lambda}_{\bom^{\pr\pr}} }^{-\frac p 2} \pa{1 + tV_{\bom^{\pr\pr}} }^{-\frac p 4}}}
\notag  \\ \notag
&=  \E_{\bom^{\pr\pr}} \set{\tr\set{ \chi\up{2}_{0} \pa{1 + tV_{\bom^{\pr\pr}} }^{-\frac p 4}  \pa{1 + tH\up{\Lambda}_{\bom^{\pr\pr}} }^{-\frac p 2} \pa{1 + tV_{\bom^{\pr\pr}} }^{-\frac p 4} \chi\up{2}_{0}} }\\ \notag
&=  \E_{\bom^{\pr\pr}} \set{\tr \set{\pa{1 + tV_{\bom^{\pr\pr}} }^{-\frac p 4} \chi\up{2}_{0}  \pa{1 + tH\up{\Lambda}_{\bom^{\pr\pr}} }^{-\frac p 2} \chi\up{2}_{0} \pa{1 + tV_{\bom^{\pr\pr}} }^{-\frac p 4}}}\\ \notag
&\le \E_{\bom^{\pr\pr}} \set{  \pa{1 + t {u_-} \hat{\omega}_0}^{-\frac p 2}\tr \set{ \chi\up{2}_{0}  \pa{1 + tH\up{\Lambda}_{\bom^{\pr\pr}} }^{-\frac p 2} \chi\up{2}_{0} }},
\end{align}
where  we set, with $ Q:=  \set{0,1}^d \setminus \set{0} \subset\Z^d$,
\beq \label{hatomega}
\hat{\omega}_0= \sum_{q \in Q} \hat{\omega}_{0,q}; \quad \text{with} \quad  \hat{\omega}_{0,q}:= \min \set{\omega_{q +i}; \; i \in 2 \Z^d,   \abs{q +  i }_\infty=1  }.
\eeq
Note that $\abs{Q}=2^d -1$,  and $\pa{q + 2\Z^d}\cap \pa{q^\pr + 2\Z^d}=\emptyset$ if $q,q^\pr \in Q$, with $q\not= q^\pr$, so  $\set{ \hat{\omega}_{0,q}}_{q \in Q}$ are independent random variables. 

Now,  with $\Theta:= \max \set{- \essinf V_{\mathrm{per}}, 0} $,
\begin{align}\label{trtheta}
& \tr\set{ \chi\up{2}_0   \pa{1 + tH\up{\Lambda}_{\bom^{\pr\pr}} }^{-\frac p 2} \chi\up{2}_0 }\\
&\quad  \le  \set{ \sup_{E \ge0} \pa{\tfrac {1 +\Theta + E }{1 +tE}}^{\frac p 2} }
 \tr \set{\chi\up{2}_0  \pa{H\up{\Lambda}_{\bom^{\pr\pr}}+1+\Theta}^{-\frac p 2}\chi\up{2}_0 } 
  \le C_{d,\Theta} \max \set{1 ,  t ^{-\frac p 2}} , \notag
\end{align}
where we used the fact that since $p=2^{d+1} \ge 4[[\frac d 4]]$, where $ [[\frac d 4]]$ is the smallest integer $> \frac d 4$,  it follows that $ \tr\set{ \chi\up{2}_0  \pa{H\up{\Lambda}_{{\bom^{\pr\pr}}}+1+\Theta }^{-\frac p 2}\chi\up{2}_0}$ is uniformly bounded, independently of $\Lambda$
(e.g., as in  \cite[proof of Lemma~A.4]{GKduke}).

Moreover, since $p=2^{d+1} > 2(2^d -1)$,
\begin{align}\notag
&
\E_{\bom^{\pr\pr}}\set{ \pa{1 + t {u_-} \hat{\omega}_0 }^{-\frac p 2}} \le \prod_{q \in Q} \E_{\bom^{\pr\pr}}\set{ \pa{1 + t {u_-} \hat{\omega}_{0 ,q}}^{-\frac {p}{2(2^d -1)}}}\\ \notag
& \qquad \qquad  =  \prod_{q \in Q} \E_{\bom^{\pr\pr}}\set{ \max_{i \in 2\Z^d; \; \abs{q + i}_\infty=1}   \pa{1 + t {u_-} {\omega}_{q+i}}^{-\frac {p}{2(2^d -1)}}}
\\
& \qquad \qquad \le  \pa{2d \,  \E_{\omega_0}\set{\pa{1 + t {u_-} {\omega}_{0}}^{-\frac {p}{2(2^d -1)}}} }^{2^d -1}\\
& \notag
 \qquad \qquad  \le \pa{2d \,  \rho_+ \int_0^\infty \di\omega_0\, \pa{1 + t {u_-} {\omega}_{0}}^{-\frac {p}{2(2^d -1)}}}^{2^d -1}
\\ \notag
& \qquad \qquad \le  \pa{\frac {2d (2^d -1)\,  \rho_+}{(2^d -1-\frac p 2)  t {u_-}  } }^{2^d -1}=
 C_{d}^\pr\pa{ \frac{\rho_+}{ t {u_-}}}^{2^d -1}.
\end{align}

Thus,  we have
\beq
\beta_t \le C_{d,\Theta}^\pr \max \set{1 ,  t ^{-2^d} } \pa{ \frac{\rho_+}{ t {u_-}}}^{2^d -1},
\eeq
so it follows from \eq{estimatetrP} that
\begin{align}\label{estimatetrP1}
&\E_{\bom} \tr P\up{\Lambda}_{\bom}(I) \le \\
& \;  \tfrac {C_{d,\Theta}^\pr } {{u_-}}  \pa{1 + t E_0}^{2^{d+2}}  \, \e^{\frac {1}{2 \sqrt{t}} }  \pa{1 + 16d\sqrt{t}  }^{2d}\!\!\! \max \set{1 ,  t ^{-2^d}}  \pa{ \frac{\rho_+}{ t {u_-}}}^{2^d -1} \!\!\!\! \!\!\! \rho_+ \abs{I}  \abs{\Lambda}. \notag
\end{align}

If $E_0\le 3$, we choose  $t=\frac 1 {E_0}$, obtaining
\begin{align}\label{estimatetrP21}
\E_{\bom} \tr P\up{\Lambda}_{\bom}(I)& \le   {C^{\pr\pr}_{d,\Theta}}  \pa{ \frac{\rho_+}{ {u_-}}}^{2^d }  E_0^{2^d-d -1}   \abs{I}  \abs{\Lambda}.
\end{align}
If $E_0 >3$, we  take $t= 1 $, getting
\begin{align}\label{estimatetrP22}
\E_{\bom} \tr P\up{\Lambda}_{\bom}(I)& \le   {C^{\pr\pr\pr}_{d,\Theta}}      \pa{ \frac{\rho_+}{ {u_-}}}^{2^d }  E_0^{2^{d+2}} \abs{I}  \abs{\Lambda}.
\end{align}
Thus, for all $E_0 >0$ we have
\begin{align}\label{estimatetrP2}
\E_{\bom} \tr P\up{\Lambda}_{\bom}(I)& \le  \tfrac {C_{d,\Theta}}{u_-} \,  \pa{ \frac{\rho_+}{ {u_-}}}^{2^d-1 }\! \!\min \set{1, E_0^{2^d-d -1}}  \max \set{1 , E_0^{2^{d+2}} } \rho_+ \abs{I}  \abs{\Lambda}.
\end{align}

For $d=1$ we need to do a bit better.  In this case we redo \eq{trtheta} as follows:
\begin{align}\label{trtheta1}
 \tr \set{\chi\up{2}_0   \pa{1 + tH\up{\Lambda}_{\bom^{\pr\pr}} }^{-\frac p 2} \chi\up{2}_0 }& \le  \tr\set{ \chi\up{2}_0   \pa{1 + tH\up{\Lambda}_{\bom^{\pr\pr}} }^{-1} \chi\up{2}_0}\\
 & \le \alpha_t:=  \tr \set{\chi\up{2}_0   \pa{1 + tH\up{\Lambda}_0 }^{-1} \chi\up{2}_0}.  \notag
\end{align}
For $d=1$ the estimate \eq{estimatetrP1} now becomes
\begin{align}\label{estimatetrP1d1}
\E_{\bom} \tr P\up{\Lambda}_{\bom}(I)  \le   \tfrac {C_{1,\theta}} {{u_-}}  \pa{1 + t E_0}^{8}  \, \e^{\frac {1}{2 \sqrt{t}} }  \pa{1 + 16\sqrt{t} }^{2} \alpha_t \pa{ \frac{\rho_+}{ t {u_-}}}   \rho_+ \abs{I}  \abs{\Lambda} ,
\end{align}
and thus \eq{estimatetrP2} becomes
\begin{align}\label{estimatetrP2d1}
\E_{\bom} \tr P\up{\Lambda}_{\bom}(I)& \le  \tfrac {C_{1,\Theta}}{u_-} \,  \,   \frac{\rho_+}{ {u_-}}  \gamma_1(E_0)   \max \set{1 ,  E_0^{8 }} \rho_+ \abs{I}  \abs{\Lambda}.
\end{align}
where $\gamma_1(E_0)\le 1$ and  $\lim_{E_0\to 0} \gamma_1(E_0)=0$ uniformly in $\Lambda$ large.

This proves (i). To prove (ii), we now assume $\delta_- \ge 1$. We proceed as in the proof of (i), with  $\bom^\pr=\bom$ and $\bom^{\pr\pr}=\{ \omega_j \}_{j\notin
\Z^d}=\emptyset$, that is $V_{\bom^{\pr\pr}}=0$ and $H_{\bom^{\pr\pr}}=H_0$.  We also now fix $p=2[[\frac d 4]]$. Then \eq{trPT1} yields  \eq{trPT} with  $Q_1=   \pa{1 + t E_0}^{2p}$ and 
$T_{j,k}^{(\Lambda)}= {\hat{\chi}\up{1}}_j   \pa{1 + tH\up{\Lambda}_{0} }^{-2p} {\hat{\chi}\up{1}}_k  $. Proceeding as in \eq{det1}-\eq{longest} gives \eq{uniftracesum} with 
\beq
Q_2= \beta_t\up{0} \, \e^{\frac {1}{ 4\sqrt{t}} }  \pa{\tfrac 4 3}^{p} 
\pa{1 + {32d\sqrt{t}}  }^{2d}, 
\eeq
where, as in \eq{trtheta},
\beq
\beta_t\up{0}:= \norm{\chi\up{1}_{0} \pa{1 + tH\up{\Lambda}_{0} }^{-p} \chi\up{1}_{0}}_1 \le C_{d,\Theta} \max \set{1 ,  t ^{- p }}\le C_{d,\Theta}.
\eeq
We now set $t=1$, obtaining \eq{Q12} and \eq{KW12}. 
\end{proof}


\subsection{A Wegner estimate with $\omega_0$ fixed}\label{subsectWegner0}

Let $\Upsilon=\Lambda_L(0)$  or $\R^d$.  Given $\tau \in \R$, we consider (recall $u_0=u$)
\beq
H_{(\bom\up{0}, \tau)}\up{\Upsilon}=H_{(\bom\up{0},\omega_0= \tau)} \up{\Upsilon} = H_{\bom}\up{\Upsilon} + (\tau- \omega_0) u .
\eeq

\begin{lemma}\label{lemWegner0} Let $H_\bom$ be an Anderson Hamiltonian, $E_0 >0$.
Given $\tau\in\R$, there exists a constant $\widetilde{K}_W=\widetilde{K}_W(d,u,V_{\mathrm{per}},E_0, \tau)$, such that for any interval $I\subset [0,E_0]$ and finite box $\Lambda=\Lambda_L(0)$  we have
\begin{equation} \label{Wegnersans0}
 \E_{\bom\up{0}} \set{\tr P_{(\bom\up{0}, \tau)}\up{\Lambda} (I)} \le \widetilde{K}_W \rho_+\abs{I}\abs{ \Lambda} .
\end{equation}
Moreover, if $\delta_- \ge 2$, we have 
\beq \label{KW123}
\widetilde{K}_W \le  C_{d,V_{\mathrm{per}},\tau} \pa{1 +  E_0}^{2[[\frac d 4]]}.
\eeq
\end{lemma}

\begin{proof}
We will show that  \cite[Proof of Theorem~1.3]{CHK2} can be modified to yield the proposition. To do so, we introduce the background potential
\begin{equation}
H_1 := H_0 + \tau \sum_{j\in 2\Z^d} u_j = -\Delta + V_{\mathrm{per}}\up{2},
\end{equation}
where   $V_{\mathrm{per}}\up{2}= V_{\mathrm{per}}+  \tau \sum_{j\in 2\Z^d} u_j$ 
 is a $2\Z^d$-periodic potential.
It follows that
\begin{equation}
H_{(\bom\up{0}, \tau)} = H_1 + V_{\bom\up{0}}({\tau}), \quad \text{with} \quad  V_{\bom\up{0}}({\tau}):= \sum_{j\in (2\Z)^d\setminus\{0\}} (\omega_j-\tau) u_j + \sum_{j\in  \Z\setminus(2\Z)^d} \omega_j u_j .
\end{equation}
The main point is that the single-site potential $u_0=u$ does not appear in the sum, but all the other  $u_j$'s appear with a random coefficient.

To prove \eq{Wegnersans0} with no conditions on $ \delta_-$, we proceed as in \cite[Section~2]{CHK2}.  We take an interval  $I\subset [0,E_0]$,    write $\tilde{I}= [0, E_0 +1]$; $I$ and $\tilde{I}$ replace the intervals $\Delta$ and $\tilde{\Delta}$ in \cite{CHK2}. The potential $V_{\Lambda}$ in \cite[Eq.~(2.7)]{CHK2} is replaced by $ V_{\bom\up{0}}\up{\Lambda}({\tau})$, which only involves the random variables $\bom\up{0}$.  As a consequence, the sum in \cite[Eq.~(2.10)]{CHK2} runs over indices $i,j\in\widetilde{\Lambda}\setminus\{0\}$. The spectral averaging in \cite[Eq.~(2.13)]{CHK2} can thus be performed with respect to the random variables $ {\bom\up{0}} $. Similarly for \cite[Eq.~(2.18)]{CHK2}, since $\tilde{K}(n)_{i_1,j_n}$ of \cite[Eq.~(2.17)]{CHK2} is now constructed only with the single-site potentials  $u_j$'s present in $ V_{\bom\up{0}}\up{\Lambda}({\tau})$, that is, $u_j$ with  $j\in \widetilde{\Lambda}\setminus\{0\}$. We thus get the analog of \cite[Eq.~(2.20)]{CHK2}, with $M_0=M_\rho + \abs{\tau}$, namely, with $P_1(B)=\chi_{B}(H_1)$,
\beq
 \E_{\bom\up{0}}\set{ \tr\set{ P_{(\bom\up{0}, \tau)}\up{\Lambda} (I)  P_1\up{\Lambda}(\R \setminus \tilde{I})}} \le K_1  \rho_+ \abs{I} \abs{\Lambda} ,
\eeq
for an appropriate constant $K_1$.

It remains to bound $ \E_{\bom\up{0}}\set{ \tr\set{ P_{(\bom\up{0}, \tau)}\up{\Lambda} (I)  P_1\up{\Lambda}(\tilde{I})}}$.
For this purpose, we set
\beq
\tilde{V}_1 = \sum_{j\in  (\e_1+2\Z^d)}  u_j,
\eeq
where $\e_1=(1,0,0,\ldots,0)\notin 2\Z^d$, 
 we use $H_1$ and $\tilde{V}_1\up{\Lambda}$, the restriction of $\tilde{V}_1$ to $\Lambda$,
 instead of $H_0$ and $\tilde{V}_\Lambda = \sum_{j\in \Z^d\cap\Lambda}  u_j$, in  the crucial estimate \cite[Eq.~(2.1)]{CHK2}) . Since $H_1$ and $\tilde{V}_1$ are both $2\Z^d$-periodic, we have, by \cite[Proposition~1.3]{CHK1} (see also \cite[Theorem~2.1]{CHK2}) the equivalent of (\cite[Eq.~(2.1)]{CHK2}),
\begin{equation}
 P\up{\Lambda}_1(\tilde{I})  \tilde{V}\up{\Lambda}_1 P\up{\Lambda}_1(\tilde{I}) \ge C(E_0, u, V_{\mathrm{per}},\tau)P_1\up{\Lambda}(\tilde{I}) ,
\end{equation}
with a constant $C(E_0, u, V_{\mathrm{per}},\tau) >0 $.
Since 
\begin{equation}
 \tilde{V}_1 \le \tilde{V}_{0^\perp} :=  \sum_{j\in  \Z^d\setminus\{0\}}  u_j ,
\end{equation}
it follows that
\begin{equation}
 P\up{\Lambda}_1(\tilde{I})  \tilde{V}_{0^\perp}\up{\Lambda} P\up{\Lambda}_1(\tilde{I}) \ge C(E_0, u, V_{\mathrm{per}},\tau)P_1\up{\Lambda}(\tilde{I}) .
\end{equation}
As a consequence, we get \cite[Eq.~(2.21)]{CHK2} with $\tilde{V}_\Lambda$ replaced by $ \tilde{V}_{0^\perp}\up{\Lambda}$, and hence we obtain the analogous of  \cite[Eq.~(2.31)]{CHK2}: 
\beq
\E_{\bom\up{0}}\set{\tr \set{ P\up{\Lambda}_1(\tilde{I})  \tilde{V}_{0^\perp}\up{\Lambda}  P_{(\bom\up{0}, \tau)}\up{\Lambda} (I) \tilde{V}_{0^\perp}\up{\Lambda} P\up{\Lambda}_1(\tilde{I}) }}
  \le K_2   \rho_+ \abs{I} \abs{\Lambda},
  \eeq
  for an appropriate constant $K_2$.
  
  The desired \eq{Wegnersans0} now follows, as the analogue of \cite[Eq.~(2.32)]{CHK2}.
  
   If   $ \delta_- \ge 2$, we  have
   \begin{equation}\label{cover2119}
 \sum_{j\in \pa{\pa{j_0 +  \Z^d}\setminus \{0\}}\cap \Lambda}u_j \ge  u_- \chi_\Lambda, \end{equation}
   so we can apply the proof of  Lemma~\ref{thmWsmall} (ii) to the random operator $H_{\bom\up{0}, \tau}$ getting \eq{Wegnersans0} with \eq{KW123}.
 \end{proof}

\section{The Minami estimate}\label{secMinami}

Theorem~\ref{thmminami}  follows by combining Lemma~\ref{thmWsmall}(i) and the following lemma.

\begin{lemma} \label{lemMinami}  Let  $H_\bom$ be an Anderson Hamiltonian   with a uniform-like distribution $\mu$.  Let $E_0>0$ and suppose the Wegner estimate \eq{Wegner} holds for all
 intervals $I \subset [0,E_0]$  with a constant $K_W$ such that
 \beq\label{condK}
2 K_W {U_+} \tfrac{\rho_+}{\rho_-} \le 1.
\eeq 
Then there exists a constant $K_M=K_M(u,\rho_\pm,M_\rho,E_0,d)$ such that the Minami estimate \eq{Minami} holds for all intervals  $I\subset [0,E_0]$.

If $ \delta_- \ge 2$, we have the estimate
\beq\label{KMest}
K_M \le  C_{d,V_{\mathrm{per}},M_\rho} \pa{1 +  E_0}^{4[[\frac d 4]]}.
\eeq
\end{lemma}

\begin{proof} Let $\Lambda$ be a finite box. 
It follows from  \eq{Wegner}
that $\E_\bom\set{ \tr P_\bom^{(\Lambda)}(\{c\})}=0$ for any $c \in \R$. Thus    we may
take all bounded  intervals to be of the form $]a,b]$. For such an interval we  modify Lemma~\ref{lemkey} as follows:
Given $\delta >0$ small, we pick a nonincreasing function  ${h} \in C^\infty(\R)$, such that  ${h}(t)=1 $ for $ t\le 0$ and ${h}(t)= 0$ for $t \ge \delta$.  Note that  $0\le {h}\le 1$, $h^\pr \le 0$,  $\supp {h}^\prime \subset [0,\delta]$, $ \int_\R \di t\,  {h}^\pr(t) = - 1$, and we can choose  ${h}$  so $\abs{{h}^\prime} \le \frac 2 \delta$.  Given  $c \in \R$, we set
${h}_c(t)={h}(t-c)$, and note that   ${h}_{c-\delta} \le \chi_{]-\infty,c]} \le  h_c$.
We let $I=]a,b]$, $I_\delta=]a-\delta,b+ \delta]$. 
Using $h$, we rework \eqref{decomp} in the following way.  Given $j \in\widetilde{\Lambda}$ and  $\tau \ge M_\rho$, we have
\begin{align}\notag
& \tr  P\up{\Lambda}_{\bom}\pa{I} \le   \tr {h}_b\pa{H\up{\Lambda}_{\bom}} - \tr {h}_{a- \delta}\pa{H\up{\Lambda}_{\bom}}
\\  \label{decomp2}
&\quad  \le  \set{\tr {h}_b\pa{H\up{\Lambda}_{(\bom_j^\perp,\omega_j=0)}} -\tr {h}_b\pa{H\up{\Lambda}_{(\bom_j^\perp,\omega_j=\tau)}}} \\ \notag
& \qquad \qquad  \qquad \qquad +\set{\tr {h}_b\pa{H\up{\Lambda}_{(\bom_j^\perp,\omega_j=\tau)}} -  \tr {h}_{a- \delta}\pa{H\up{\Lambda}_{(\bom_j^\perp,\omega_j=\tau)}} }\\ 
\notag
& \quad \le  \set{\tr {h}_b\pa{H\up{\Lambda}_{(\bom_j^\perp,\omega_j=0)}} -\tr {h}_b\pa{H\up{\Lambda}_{(\bom_j^\perp,\omega_j=\tau)}}}  +\tr  P\up{\Lambda}_{(\bom_j^\perp,\omega_j=\tau)}\pa{I_\delta}.
\end{align}

We now fix $\tau=M_\rho$ and  use the Birman-Solomyak formula (cf.\  \cite{Si2}) as in \cite[Eqs.~(7)-(8)]{CHK3}, 
plus the hypothesis \eq{unifdist}, obtaining
\begin{align}\notag
 \xi\up{\Lambda}_{b,\tau} (\bom_j^\perp) &:={\tr {h}_b\pa{H\up{\Lambda}_{(\bom_j^\perp,\omega_j=0)}} -\tr {h}_b\pa{H\up{\Lambda}_{(\bom_j^\perp,\omega_j=\tau)}}} \\ \label{xib}
 &  = - \int_0^\tau \di s \, \tr \set{\sqrt{u_j}{h}^\pr_b\pa{H\up{\Lambda}_{(\bom_j^\perp,\omega_j=s)}}\sqrt{u_j}}\\ \notag
 & \le \tfrac 2 \delta  \int_0^\tau \di s \, \tr \set{\sqrt{u_j} P\up{\Lambda}_{(\bom_j^\perp,\omega_j=s)}\pa{]b,b+ \delta]}\sqrt{u_j}}\\
  &\le \tfrac 2 {\delta \rho_-} \int \di s \, \rho(s) \tr \set{\sqrt{u_j} P\up{\Lambda}_{(\bom_j^\perp,\omega_j=s)}\pa{]b,b+ \delta]}\sqrt{u_j}}.\notag
\end{align}
Note that $\xi\up{\Lambda}_{b,\tau} (\bom_j^\perp)$ is closely related to the spectral shift function associated to the pair $H\up{\Lambda}_{(\bom_j^\perp,\omega_j=0)}$ and $H\up{\Lambda}_{(\bom_j^\perp,\omega_j=\tau)}$.

Now let us fix $E_0>0$,  let $I=]a,b] \subset [0,E_0[$, and consider $\delta>0$ such that $b+\delta \le E_0$, so $I_\delta \subset [0,E_0]$.  If $\tr P_\bom^{(\Lambda)}(I)\ge 1$, it follows from   \eq{traceclassexp}  that
\begin{align}\label{bddecomp}
& \pa{\tr P_\bom^{(\Lambda)}(I)}\pa{\tr P_\bom^{(\Lambda)}(I)-1 }\le  {Q_1}  \sum_{j \in \widetilde{\Lambda}} 
\tr \set{\sqrt{u_j}  P_\bom^{(\Lambda)}(I)\sqrt{u_j } \,S_j^{(\Lambda)} } \pa{\tr P_\bom^{(\Lambda)}(I)-1 },
\end{align}
so, using \eq{decomp2}  and  \eq{xib}, we get
\begin{align}
& \pa{\tr P_\bom^{(\Lambda)}(I)}\pa{\tr P_\bom^{(\Lambda)}(I)-1 }\le  {Q_1}  \sum_{j \in \widetilde{\Lambda}} \set{
\pa{\tr \set{\sqrt{u_j}  P_\bom^{(\Lambda)}(I)\sqrt{u_j } \,S_j^{(\Lambda)} } } \Phi\up{\Lambda}_{b,\tau} (\bom_j^\perp) }, \label{decompMinami}
\end{align}
where for each $j \in \widetilde{\Lambda}$
\beq
\Phi\up{\Lambda}_{b,\tau} (\bom_j^\perp):=  \pa{\xi\up{\Lambda}_{b,\tau} (\bom_j^\perp) -1 }  +   \tr  P\up{\Lambda}_{(\bom_j^\perp,\tau)}\pa{I_\delta}   
\eeq
is independent of the random variable $\omega_j$.  If  $\tr P_\bom^{(\Lambda)}(I)< 1$ , we have $ P_\bom^{(\Lambda)}(I)=0$, and hence we also have  \eq{decompMinami}.

Thus, if we now take the expectation  in  \eq{decompMinami}, use \eq{sa1} and \eq{uniftracesum2}, we get
\begin{align}\notag
\E\set{  \pa{\tr P_\bom^{(\Lambda)}(I)}\pa{\tr P_\bom^{(\Lambda)}(I)-1 }}&\le Q_1Q_2 \rho_+ \abs{I}  \sum_{j \in \widetilde{\Lambda}}  \E_{\bom_j^\perp} \set{\Phi\up{\Lambda}_{b,\tau} (\bom_j^\perp)}\\
&= Q_1Q_2 \rho_+ \abs{I} \abs{\Lambda}   \E_{\bom_{k}^\perp} \set{\Phi\up{\Lambda}_{b,\tau} (\bom_{k}^\perp)} \label{Q1Q2}
\end{align}
for any $k   \in \widetilde{\Lambda}$.

We will now  estimate $\E_{\bom_{k}^\perp}\set{ \Phi\up{\Lambda}_{b,\tau} (\bom_{k}^\perp)}$.   It follows from \eq{xib} and \eq{Wegner}  that, if we have \eq{condK}, 
\begin{align}\notag
 &\E_{\bom_{k}^\perp}\set{\xi\up{\Lambda}_{b,\tau} (\bom_k^\perp) }\le  \tfrac 2 {\delta \rho_-}   \E_{\bom} \set{\tr \set{\sqrt{u_k} P\up{\Lambda}_{\bom}\pa{]b,b+ \delta]}\sqrt{u_k}}}
 \\ \label{avessf}
 &\qquad =  \tfrac 2 {\delta \rho_- \abs{\Lambda}}   \E_{\bom} \set{  \sum_{j \in \widetilde{\Lambda}} \tr \set{\sqrt{u_j} P\up{\Lambda}_{(\bom)}\pa{]b,b+ \delta]}\sqrt{u_j}}}\\
 &\qquad \le  \tfrac {2  U_+}{\delta \rho_-\abs{\Lambda}}  \E_{\bom} \set{\tr  P\up{\Lambda}_{\bom}\pa{]b,b+ \delta]}}\le   2 K_W U_+\tfrac { \rho_+ }  { \rho_-}\le 1. \notag
\end{align}
In this case, we have
\begin{equation} \label{Wegnersans0k}
\E_{\bom_{k}^\perp}\set{ \Phi\up{\Lambda}_{b,\tau} (\bom_{k}^\perp)} \le  \E_{\bom_{k}^\perp}\set{ \tr  P\up{\Lambda}_{(\bom_{k}^\perp,\tau)}\pa{I_\delta}  } \le \widetilde{K}_W \rho_+(\abs{I}+ 2\delta)\abs{ \Lambda} ,
\end{equation}
where we used  Lemma~\ref{lemWegner0}, where $\widetilde{K}_W=\widetilde{K}_W (d,u,V_{\mathrm{per}},E_0, M_\rho) $.

Combining \eq{Q1Q2} and \eq{Wegnersans0k} we get
\begin{align}
\E\set{  \pa{\tr P_\bom^{(\Lambda)}(I)}\pa{\tr P_\bom^{(\Lambda)}(I)-1 }}&\le 
Q_1Q_2  \widetilde{K}_W \abs{I}  (\abs{I}+ 2\delta)\pa{\rho_+\abs{ \Lambda}}^2.
\end{align}
Letting  $\delta\to 0$ we get \eq{Minami} with $K_M= Q_1Q_2  \widetilde{K}_W$.

If $ \delta_- \ge 2$, the estimate
\eq{KMest}  follows from \eq{Q12} and \eq{KW123}.
\end{proof}





\section{Poisson statistics}
\label{sectPoisson}

In this section we prove  Theorem~\ref{thmPoisson}(a). 

 Let $H_\bom$ be an Anderson Hamiltonian, and suppose  $\mathcal{I} $ is  an open interval
such  that for all large boxes $\Lambda$ the estimate \eq{Minami} holds for any
interval $I \subset \mathcal{I}$ with $\abs{I}\le \delta_0$, for some $\delta_0 >0$, with some constant $K_M$. (We will assume that a given $\Lambda$ is large enough.)  Recall we have \eq{Wegner} for these intervals with some constant $K_W$. 

  Let  $\cE \in\mathcal{I}\cap  \Xi^{\text{CL}}$ be such that the IDS $N(E)$ is differentiable at
$\cE$ with $n(\cE):=N^\pr(\cE) >0$.  
It follows from \eq{Wegner} that  we then have 
 \beq  \label{intn}
0 < n(\cE)\le K_W \rho_+.
 \eeq 
We fix an open interval $\mathcal{I}_1$  such that  $\cE \in \mathcal{I}_1\subset  \overline{\mathcal{I}_1}\subset \mathcal{I}\cap  \Xi^{\text{CL}}$. Note that for each  bounded  Borel set $B \subset \R$ there exists $c_B=c_{B,\cE, \mathcal{I}_1} <\infty$ such that $\cE + \abs{\Lambda}^{-1} B \subset \mathcal{I}_1$ and  
$\abs{\cE + \abs{\Lambda}^{-1} B} \le \delta_0$ if $\abs{\Lambda}\ge c_B$.
 The point process  $\xi_{\bom}^{(\Lambda)}=\xi_{\cE,\bom}^{(\Lambda)}$ (cf.\  \eq{defxi0})
has an  intensity measure given by 
$\nu^{(\Lambda)}(B):=\E\,  \xi_{\bom}^{(\Lambda)}(B)$ for a Borel set $B \subset \R$;  
it follows from \eq{Wegner} that,
\beq \label{int0}
\nu^{(\Lambda)}(B)\le K_W \rho_+ \abs{B} \quad\text{for all $\Lambda$ with    $\abs{\Lambda}\ge  c_B$}.
\eeq 
 
We start with  the same general strategy used in \cite{Mo,Mi}.
We fix $a \in ]0,1[$, and divide $\Lambda=\Lambda_L(0)$ into $M_L $ boxes   $\Lambda^{(m)}=\Lambda_\ell(k_m)$ of side $\ell \approx L^a$, $\ell \in 2\N$, centered at  $k_m \in \Lambda \cap (2\Z^d)$; note $M_L= \frac {\abs{\Lambda_L}}{\abs{\Lambda_\ell}} \approx L^{(1 -a)d}$.
For each $m=1,2,\ldots,M_L$ we define point processes
\beq \label{defxim}
\xi_{\bom}^{(\Lambda,m)}(B):= \tr P_{\bom}^{(\Lambda^{(m)})}(\cE + \abs{\Lambda}^{-1} B) \quad \text{for a Borel set $B \subset \R$}.
\eeq
Note that  $\set{\xi_{\bom}^{(\Lambda,m)}}_{m=1,2,\ldots,M_L }$ are independent, identically distributed point processes, each with intensity measure (using  \eq{Wegner})
\beq \label{int1}
\nu^{(\Lambda,m)}(B):=\E\,  \xi_{\bom}^{(\Lambda,m)}(B)\le K_W \rho_+ \abs{B}M_L^{-1} \quad\text{for all $\Lambda$ with    $\abs{\Lambda}\ge  c_B$}.\eeq
 We consider their superposition,  the point process
\beq
\widetilde{\xi}_{\bom}^{(\Lambda)}:= \sum_{m=1}^{M_L} \xi_{\bom}^{(\Lambda,m)},
\eeq
with intensity measure
\beq \label{int2}
\widetilde{\nu}^{(\Lambda)}(B):=\E\, \widetilde{\xi}_{\bom}^{(\Lambda)} (B)\le K_W \rho_+ \abs{B} \quad\text{for all $\Lambda$ with    $\abs{\Lambda}\ge  c_B$}.
\eeq
We will prove that  $\widetilde{\xi}_{\bom}^{(\Lambda)}\approx {\xi}_{\bom}^{(\Lambda)}$ as $L \to \infty$, and that $\widetilde{\xi}_{\bom}^{(\Lambda)}$ converges weakly, as $L \to \infty$, to the Poisson point process $\xi$  with intensity measure  $\nu(B):=\E\,  \xi(B)= n(\cE)\abs{B}$.   But here we must use different methods from  \cite{Mo,Mi}.  

So let $\theta_{\bom}^{(\Lambda)}=\theta_{\cE,\bom}^{(\Lambda)}$ be the random measure
defined in \eq{deftheta};
its  intensity measure is
\beq \label{thetaint}
\eta^{(\Lambda)}(B):=\E\,  \theta_{\bom}^{(\Lambda)}(B) = \abs{\Lambda} \eta(\cE + \abs{\Lambda}^{-1} B),
\eeq
where $\eta$ is the  density of states  measure, given in \eq{dsm}.
 It again  follows from \eq{Wegner} that 
\beq \label{int5}
\eta^{(\Lambda)}(B)\le K_W \rho_+ \abs{B}\quad\text{for all $\Lambda$ with    $\abs{\Lambda}\ge  c_B$}.
\eeq

We start with a lemma. Given a measure $\eta$ on $\R$,  we write $\eta(f):= \int_{\R} f \, \di \eta$ for  suitable functions $f$, say, $f \in \cF_{b,K}$, the collection of bounded Borel functions on $\R$ vanishing outside a compact interval. 
  It follows from \eq{defxi0}  that for  all $ f \in\cF_{b,K}$ we have
\beq
{\xi}_{\bom}^{(\Lambda)}(f)=  \tr f_{\Lambda}(H_{\bom}^{(\Lambda)}), \quad\text{where} \quad  f_{\Lambda}(E):= f\pa{ \abs{\Lambda}(E - \cE)}, 
    \eeq
with similar expressions for $\widetilde{\xi}_{\bom}^{(\Lambda)}(f)$,  $\xi_{\bom}^{(\Lambda,m)}(f)$, and $\theta_{\bom}^{(\Lambda)}(f).$

\begin{lemma} \label{lemEconv} For all  $ f \in \cF_{b,K}$ we have
\beq \label{limL}
\lim_{L \to \infty} \E \, \abs{{\xi}_{\bom}^{(\Lambda)}(f)-\widetilde{\xi}_{\bom}^{(\Lambda)} (f)}=0
\eeq
and
\beq \label{limL3}
\lim_{L \to \infty} \E \, \abs{{\xi}_{\bom}^{(\Lambda)}(f)-{\theta}_{\bom}^{(\Lambda)} (f)}=0.
\eeq
\end{lemma}

\begin{proof}
In view of \eq{int0}, \eq{int2}, and \eq{int5}, it suffices to prove \eq{limL} and \eq{limL3} for  $ f \in C_{K}^\infty(\R)$, since  
 $ \set{f \in C_{K}^\infty(\R);  \ \supp f \subset J}$ is dense in $ \L^1(J,\di E)$ for any interval $J$.

So let $ f \in C_{K}^\infty(\R)$.  To prove   \eq{limL}, we set $\ell^\pr \approx \ell - \sqrt{\ell}$,  $\Lambda^{(m,\pr)}=\Lambda_{\ell^\pr}(k_m)$,  and  $\Lambda^{(m,\pr\pr)}=\Lambda_{\ell}(k_m)\setminus \Lambda_{\ell^\pr}(k_m)$.  Using  $\chi_{\Lambda}= \sum_{m=1}^{M_L} \chi_{\Lambda^{(m)}}$,   we get
\begin{align}
&{\xi}_{\bom}^{(\Lambda)}(f)-\widetilde{\xi}_{\bom}^{(\Lambda)} (f)=
\sum_{m=1}^{M_L}  \pa{ \tr\set{\chi_{\Lambda^{(m)}} f_{\Lambda}(H_{\bom}^{(\Lambda)})\chi_{\Lambda^{(m)}} } -  \tr   f_{\Lambda}(H_{\bom}^{(\Lambda^{(m)})})   }\\ \notag
&\quad = \sum_{m=1}^{M_L}  \pa{ \tr\set{\chi_{\Lambda^{(m,\pr)}} f_{\Lambda}(H_{\bom}^{(\Lambda)})\chi_{\Lambda^{(m,\pr)}} } -  \tr \set{\chi_{\Lambda^{(m,\pr)}}  f_{\Lambda}(H_{\bom}^{(\Lambda^{(m)})} )\chi_{\Lambda^{(m,\pr)}}  }}  \\ \notag
&\qquad \qquad + \sum_{m=1}^{M_L}  \pa{ \tr\set{\chi_{\Lambda^{(m,\pr\pr)}} f_{\Lambda}(H_{\bom}^{(\Lambda)})\chi_{\Lambda^{(m,\pr\pr)}} } -  \tr \set{\chi_{\Lambda^{(m,\pr\pr)}}  f_{\Lambda}(H_{\bom}^{(\Lambda^{(m)})} )\chi_{\Lambda^{(m,\pr\pr)}}  }}.
\end{align}
We now use the fact that the expectation is invariant under translations in the torus to get, for any $m$, 
\begin{align}\notag
& \E \, \abs{{\xi}_{\bom}^{(\Lambda)}(f)-\widetilde{\xi}_{\bom}^{(\Lambda)} (f)}\\
& \,  \le 
 M_L \E \abs{ \tr\set{\chi_{\Lambda^{(m,\pr)}} f_{\Lambda}(H_{\bom}^{(\Lambda)})\chi_{\Lambda^{(m,\pr)}} } -  \tr \set{\chi_{\Lambda^{(m,\pr)}}  f_{\Lambda}(H_{\bom}^{(\Lambda^{(m)})}) \chi_{\Lambda^{(m,\pr)}}  }} \label{mainterm}\\
 & \; + M_L \E \abs{ \tr\set{\chi_{\Lambda^{(m,\pr\pr)}} f_{\Lambda}(H_{\bom}^{(\Lambda)})\chi_{\Lambda^{(m,\pr\pr)}} } -  \tr \set{\chi_{\Lambda^{(m,\pr\pr)}}  f_{\Lambda}(H_{\bom}^{(\Lambda^{(m)})}) \chi_{\Lambda^{(m,\pr\pr)}}  }} .\label{boundaryterm}
 \end{align} 
 
It  follows from  the Wegner estimate \eq{Wegner} that
\begin{align}
& M_L \E \abs{ \tr\set{\chi_{\Lambda^{(m,\pr\pr)}} f_{\Lambda}(H_{\bom}^{(\Lambda)})\chi_{\Lambda^{(m,\pr\pr)}} }}\le M_L \frac {\abs{\Lambda^{(m,\pr\pr)}}}{\abs{\Lambda}}
 \E \tr\set{ \abs{ f_{\Lambda}}(H_{\bom}^{(\Lambda)})}\\
 & \quad  \le M_L \frac {\abs{\Lambda^{(m,\pr\pr)}}}{\abs{\Lambda}}K_W \rho_+ \abs{\Lambda}
 \int_\R  \abs{ f_{\Lambda}}(E) \, \di E = \frac {\abs{\Lambda^{(m,\pr\pr)}}}{\abs{\Lambda^{(m)}}}  K_W \rho_+ \norm{f}_1 .
\notag
\end{align}
Similarly, 
\begin{align}
& M_L \E \abs{ \tr\set{\chi_{\Lambda^{(m,\pr\pr)}} f_{\Lambda}(H_{\bom}^{(\Lambda^{(m)})})\chi_{\Lambda^{(m,\pr\pr)}} }}\le M_L \frac {\abs{\Lambda^{(m,\pr\pr)}}}{\abs{\Lambda^{(m)}}}
 \E \tr\set{ \abs{ f_{\Lambda}}(H_{\bom}^{(\Lambda^{(m)})})}\\
 & \qquad \qquad   \le M_L \frac {\abs{\Lambda^{(m,\pr\pr)}}}{\abs{\Lambda^{(m)}}}K_W \rho_+ \abs{\Lambda^{(m)}}
 \int_\R  \abs{ f_{\Lambda}}(E) \, \di E = \frac {\abs{\Lambda^{(m,\pr\pr)}}}{\abs{\Lambda^{(m)}}}  K_W \rho_+ \norm{f}_1 .
\notag
\end{align}
Since
\beq
 \frac {\abs{\Lambda^{(m,\pr\pr)}}}{\abs{\Lambda^{(m)}}} \approx \frac {\ell^{d-1} \sqrt{\ell}}{\ell^d}= \frac 1 { \sqrt{\ell}}\approx \frac 1 {L^\frac a 2}\to 0 \quad \text{as} \quad L \to \infty,
\eeq
the term in \eq{boundaryterm} goes to $0$ as $L \to \infty$.

To finish the proof of \eq{limL} we need to show that the term in \eq{mainterm} also goes to $0$ as $L \to \infty$.  To do that we will use that  $\overline{\cI_1} \subset \Xi^{\text{CL}}$,   the   Helffer-Sj\"ostrand formula for {smooth}
functions of self-adjoint operators, and  estimates on Schr\"odinger operators.

Given a box $\Lambda$, we identify   $\L^2(\Lambda)$  with the subspace of $\L^2(\R^d)$ consisting of functions vanishing outside $\Lambda$.   Given a function  $\phi \in C_{K}^\infty(\R)$, we let $W(\phi)$ to be the closure of the local first order differential operator $[\Delta, \phi]$ on  $C_{K}^\infty(\R)$.  We set $\chi_\phi:= \chi_{\supp  \phi}$, $\chi_{\nabla \phi}:= \chi_{\supp \nabla \phi}$. and note
that  $W(\phi)= \chi_{\nabla \phi} W(\phi)= W(\phi)\chi_{\nabla \phi} =\chi_{\nabla \phi}W(\phi)\chi_{\nabla \phi}$.  We recall that if $\supp  \phi \subset \Lambda^\circ$, the interior of $\Lambda$, which here  may be either  a finite box or $\R^d$, 
we have 
\beq \label{Wphibound}
\norm{\pa{H_{\bom}^{(\Lambda)} +1}^{-\frac 12 } W(\phi)}=\norm{W(\phi)\pa{H_{\bom}^{(\Lambda)} +1}^{-\frac 12 } } \le C_\phi:=C_1 \pa{\norm{\Delta \phi}_\infty + \norm{\nabla \phi}_\infty},
\eeq
where $C_1$ depends only on $d$.  We also recall that for all $x \in \Lambda $ we have 
\beq \label{normpd}
\norm {\chi_{\Lambda_1(x)} {\pa{H_{\bom}^{(\Lambda)} +1}^{-1}}}_{p_d}\le C_2 < \infty \quad \text{with} \quad p_d = [\tfrac d 2] + 1,
\eeq
the constant $C_2$ being independent of $x$ and  $\Lambda$ for $L \ge  2$ (cf.\  \cite[Eqs. (130)-(136)]{KKS}). 

We now recall
the   Helffer-Sj\"ostrand formula;  cf.\ \cite[Appendix B]{HuS} for details.   Given  $ g\in C^\infty(\R)$ and $m\in \N$, we set  \begin{equation} \label{sdfn}
  \hnorm{g}_m := \sum_{r=0}^m \int_{\mathbb{R}}\!\mathrm{d}u\;
  |g^{(r)}(u)|\,(1 + \abs{u}^{2})^{\frac {r-1} 2}  .
\end{equation}
If $ \hnorm{g}_m < \infty$ with $m \ge 2$, then for any self-adjoint operator
$K$ we have
\begin{equation}\label{HS}
  f (K) = \int_{\R^{2}} \!\di \tilde{g}(z) \, (K-z)^{-1} ,
\end{equation}
where the integral converges absolutely in operator norm.  Here $z= x + i y$,
$\tilde{g}(z)$ is an \emph{almost analytic extension} of $g$ to the complex
plane, $\di \tilde{g}(z) := \frac 1 {2\pi}\partial_{\bar{z}}\tilde{g}(z)
\,\mathrm{d} x\, \mathrm{d} y $, with $\partial_{\bar{z}}= \partial_x + i
\partial_y$, and $|\di \tilde{g}(z)| := (2\pi)^{-1}
|\partial_{\,\overline{z}}\tilde{g}(z)| \,\mathrm{d} x\, \mathrm{d} y$.
Moreover, for all $p \ge 0$ we have
\begin{equation}\label{HShigherorder}
  \int_{\R^{2}} \! |\di \tilde{g}(z)| \;\frac{1}{|\Im \, z|^p}  \le c_p
  \  \hnorm{g}_m < \infty  \quad \text{for} \quad m \ge p+1
\end{equation}
with a constant $c_{p}$.

.

 Since $ f \in C_{K}^\infty(\R)$, we have, using the Helffer-Sj\"ostrand formula, with $\Lambda=\Lambda_L$, $R_{\bom}^{(\Lambda)} (z)=
\pa{H_{\bom}^{(\Lambda)} -z}^{-1}$ and  $R_{\bom}^{(\Lambda,m)}(z)=
\pa{H_{\bom}^{(\Lambda^{(m)})} -z}^{-1}$, and taking $\phi_0 \in C_{K}^\infty(\Lambda_{\ell- 10d}(k_m))$, such that $0\le \phi_0 \le1$ and $\phi_0 \chi_{\Lambda_{\ell- 20 d}(k_m)} = \chi_{\Lambda_{\ell- 20 d}(k_m)} $, that
\begin{align}\label{TLambda}
&T_{\bom}^{(\Lambda)}:={\chi_{\Lambda^{(m,\pr)}} f_{\Lambda}(H_{\bom}^{(\Lambda)})\chi_{\Lambda^{(m,\pr)}} } -   {\chi_{\Lambda^{(m,\pr)}}  f_{\Lambda}(H_{\bom}^{(\Lambda^{(m)})}) \chi_{\Lambda^{(m,\pr)}}  }\\ \notag
&\quad = \int_{\R^{2}} \!\di \tilde{f_{\Lambda}}(z) \set{\chi_{\Lambda^{(m,\pr)}} R_{\bom}^{(\Lambda)} (z)\chi_{\Lambda^{(m,\pr)}}- \chi_{\Lambda^{(m,\pr)}} R_{\bom}^{(\Lambda,m)}(z) \chi_{\Lambda^{(m,\pr)}} } \\ \notag
&\quad = \int_{\R^{2}} \!\di \tilde{f_{\Lambda}}(z) \set{\chi_{\Lambda^{(m,\pr)}}R_{\bom}^{(\Lambda)} (z)\phi_0 \chi_{\Lambda^{(m,\pr)}}- \chi_{\Lambda^{(m,\pr)}}\phi_0  R_{\bom}^{(\Lambda,m)}(z) \chi_{\Lambda^{(m,\pr)}} }   \\ 
& \quad =   \int_{\R^{2}} \!\di \tilde{f_{\Lambda}}(z) \set{\chi_{\Lambda^{(m,\pr)}}R_{\bom}^{(\Lambda)}(z)W(\phi_0)R_{\bom}^{(\Lambda,m)}(z)\chi_{\Lambda^{(m,\pr)}} }  , \label{intHS}
\end{align}
where we used the geometric resolvent identity.

Now let us pick functions $\phi_i \in C_{K}^\infty(\R)$, $i=1,2,\ldots,2 p -1$, such that
$0\le \phi_i\le 1$, $\phi_{i}\chi_{\nabla \phi_{i-1}}=\chi_{\nabla \phi_{i-1}}$, and 
$\chi_{\phi_i} \chi_{\Lambda_{\ell- 30 d}(k_m)}=0$  for $i=1,2,\ldots,2 p-1$.  Using the resolvent identity $2p-1$ times we get
\begin{align}\label{expres}
&\chi_{\Lambda^{(m,\pr)}}R_{\bom}^{(\Lambda)}(z)W(\phi_0)\\ \notag
& \quad =
\chi_{\Lambda^{(m,\pr)}} R_{\bom}^{(\Lambda)}(z)W(\phi_{2p-1})R_{\bom}^{(\Lambda)}(z)W(\phi_{2p-2})\ldots  R_{\bom}^{(\Lambda)}(z)W(\phi_1)R_{\bom}^{(\Lambda)}(z)W(\phi_0)\\\notag 
&\quad = \set{\chi_{\Lambda^{(m,\pr)}} R_{\bom}^{(\Lambda)}(z)} \set{W(\phi_{2p-1})R_{\bom}^{(\Lambda)}(z)W(\phi_{2p-2})} \set{\chi_{\nabla \phi_{2p-2}} R_{\bom}^{(\Lambda)}(z)} \\ \notag
&\quad \quad  \quad  \times \set{W(\phi_{2p-3})R_{\bom}^{(\Lambda)}(z)W(\phi_{2p-4})}
\ldots \set{\chi_{\nabla \phi_{2}} R_{\bom}^{(\Lambda)}(z)} \set{W(\phi_{1})R_{\bom}^{(\Lambda)}(z)W(\phi_{0})}.
\end{align}
We now use that the integral in \eq{intHS} is performed over a compact domain in $\R^2$, which depends only on the function $f$, so there is constant $C_f$ such that for $z$ in the region of integration we have 
\beq \label{R2compact}
\norm{\pa{H_{\bom}^{(\Lambda)} +1}R_{\bom}^{(\Lambda)}(z)}\le \frac {C_f}{\abs{\Im z}},
\eeq
and hence, using \eq{Wphibound} and \eq{normpd}, we have
\beq  \label{Wphibound2}
 \norm{W(\phi_{i})R_{\bom}^{(\Lambda)}(z)W(\phi_{i-1})}\le \frac {C_f C_{\phi_i}  C_{\phi_{i-1}}}{\abs{\Im z}}
\eeq
and, for $B \subset \Lambda_{L^\pr}\subset \Lambda$,
\beq \label{normpd2}
\norm{\chi_{B} R_{\bom}^{(\Lambda)}(z)}_{p_d}\le  \frac {C_f C_2}{\abs{\Im z}}\abs{\Lambda_{L^\pr}}.
\eeq

We now choose $p=p_d$ as in \eq{normpd}, and note that we can choose the functions
 $\phi_i \in C_{K}^\infty(\R)$, $i=1,2,\ldots,2 p_d-1$ so that the constants $C_{\phi_i}$ are independent of $\Lambda$, say all  $C_{\phi_i}\le C_3$  From \eq{expres}, \eq{Wphibound2} and \eq{normpd2}, we get
\begin{align}\label{breakest}
&\norm{\chi_{\Lambda^{(m,\pr)}}R_{\bom}^{(\Lambda)}(z)W(\phi_0)R_{\bom}^{(\Lambda,m)}(z)\chi_{\Lambda^{(m,\pr)}} }_1 \\ \notag
&\qquad \qquad  \le  \pa{  \frac {C_f C_2}{\abs{\Im z}}\abs{{\Lambda^{(m)}}}}^{p_d}\pa{\frac {C_f C_3^2}{\abs{\Im z}}}^{p_d}\norm{\chi_{\nabla \phi_{0}}R_{\bom}^{(\Lambda,m)}(z)\chi_{\Lambda^{(m,\pr)}}}\\
&\qquad \qquad \le C_4 C^\pr_f \ell^{p_d} \abs{\Im z}^{-2 p_d}\norm{\chi_{\nabla \phi_{0}}R_{\bom}^{(\Lambda,m)}(z)\chi_{\Lambda^{(m,\pr)}}}.  \notag
\end{align}

We now use that  $\overline{\cI_1}\subset \Xi^{\text{CL}}$,  the region of complete localization for $H_\bom$. The term in \eq{mainterm} is $M_L \E\set{T_{\bom}^{(\Lambda)}}$, with $T_{\bom}^{(\Lambda)}$ as in \eq{TLambda}.  It follows  from \eq{intHS}, \eq{expres} and \eq{breakest} that for large $L$, 
\begin{align}\notag
& M_L \E\set{T_{\bom}^{(\Lambda)}}\le M_L  C_4 C^\pr_f \ell^{p_d} \int_{\R^{2}} \!\abs{\di \tilde{f_\Lambda}(z) }\, \abs{\Im z}^{-2 p_d}
\E\set{\norm{\chi_{\nabla \phi_{0}}R_{\bom}^{(\Lambda,m)}(z)\chi_{\Lambda^{(m,\pr)}}}}\\
& \quad \le  M_L  C_4C^\pr_f \ell^{p_d} \int_{\R^{2}} \!\abs{\di \tilde{f_\Lambda}(z) }\, \abs{\Im z}^{-2 p_d- \frac 4 5}
\E\set{\norm{\chi_{\nabla \phi_{0}}R_{\bom}^{(\Lambda,m)}(z)\chi_{\Lambda^{(m,\pr)}}}^{\frac 1 5}}\\\notag
& \quad\le M_L  C_4C^\pr_f \ell^{p_d + 2d}(\rho_+  + \sqrt{\rho_+}) \int_{\R^{2}} \!\abs{\di \tilde{f_\Lambda}(z) }\, \abs{\Im z}^{-2 p_d- \frac 4 5}  \mathrm{e}^{-\ell^{\frac 1 4}}\\
&\quad \le L^d  \ell^{p_d + d} \mathrm{e}^{-\ell^{\frac 1 4}}c_{2 p_d+ \frac 4 5} C_4 C^\pr_f(\rho_+  + \sqrt{\rho_+}) \hnorm{f_\Lambda}_{2p_d +2}. \notag
\end{align}
where we used \eq{eqMSAexp} and \eq{HShigherorder}. Note that $2p_d \le d+ 1$ and 
\beq
 \hnorm{f_\Lambda}_m \le C_{E_0,f,m} \abs{\Lambda}^{m-1} \quad \text{for all} \quad  m=2,3,\ldots.
\eeq
It follows that
\begin{align}
 M_L \E\set{T_{\bom}^{(\Lambda)}} \le  L^{d^2 + 3d}  \ell^{\frac {3d}2 + 1} \mathrm{e}^{-\ell^{\frac 1 4}}c_{2 p_d+ \frac 4 5} C_{f,E_0,d} (\rho_+  + \sqrt{\rho_+})\to 0 \; \text{as $L\to \infty$}.
\end{align}
Thus \eq{limL} is proven.

The proof of \eq{limL3} is similar. With $\Lambda=\Lambda_L(0)$, we set $L^\pr \approx L- \sqrt{L}$,  $\Lambda^\pr=\Lambda_{L^\pr}(0)$,  and  $\Lambda^{\pr\pr}=\Lambda \setminus \Lambda^{\pr}$. We have 
\begin{align}
&{\theta}_{\bom}^{(\Lambda)}(f)-{\xi}_{\bom}^{(\Lambda)} (f)=
\tr\set{\chi_{\Lambda} f_{\Lambda}(H_{\bom})\chi_{\Lambda} } -  \tr   f_{\Lambda}(H_{\bom}^{(\Lambda)})   \\ \notag
&\qquad \qquad = \pa{ \tr\set{\chi_{\Lambda^\pr} f_{\Lambda}(H_{\bom})\chi_{\Lambda^\pr} } -  \tr \set{\chi_{\Lambda^\pr}   f_{\Lambda}(H_{\bom}^{(\Lambda)}) \chi_{\Lambda^\pr}   }}  \\ \notag
&\qquad \qquad  \qquad \quad  + \pa{ \tr\set{\chi_{\Lambda^{,\pr\pr}} f_{\Lambda}(H_{\bom})\chi_{\Lambda^{\pr\pr}} } -  \tr \set{\chi_{\Lambda^{\pr\pr}} f_{\Lambda}(H_{\bom}^{(\Lambda)})\chi_{\Lambda^{\pr\pr}}  }},
\end{align}
and hence
\begin{align}\notag
&\E \abs{{\theta}_{\bom}^{(\Lambda)}(f) -{{\xi}_{\bom}^{(\Lambda)} (f)}}\le 
  \\ \label{mainterm2}
&\qquad \qquad =\E \abs{ \tr\set{\chi_{\Lambda^\pr} f_{\Lambda}(H_{\bom})\chi_{\Lambda^\pr} } -  \tr \set{\chi_{\Lambda^\pr}   f_{\Lambda}(H_{\bom}^{(\Lambda)}) \chi_{\Lambda^\pr}   }}  \\  \label{boundaryterm2}
&\qquad \qquad  \qquad \quad  + \E \abs{ \tr\set{\chi_{\Lambda^{,\pr\pr}} f_{\Lambda}(H_{\bom})\chi_{\Lambda^{\pr\pr}} } -  \tr \set{\chi_{\Lambda^{\pr\pr}} f_{\Lambda}(H_{\bom}^{(\Lambda)})\chi_{\Lambda^{\pr\pr}}  }}.
\end{align}

We now use the Wegner estimate \eq{Wegner} to obtain
\begin{align}
&  \E \abs{ \tr\set{\chi_{\Lambda^{\pr\pr}} f_{\Lambda}(H_{\bom}^{(\Lambda)})\chi_{\Lambda^{\pr\pr}} }}\le  \frac {\abs{\Lambda^{\pr\pr}}}{\abs{\Lambda}}
 \E \tr\set{ \abs{ f_{\Lambda}}(H_{\bom}^{(\Lambda)})}\\
 & \qquad \qquad   \le \frac {\abs{\Lambda^{\pr\pr}}}{\abs{\Lambda}}K_W \rho_+ \abs{\Lambda}
 \int_\R  \abs{ f_{\Lambda}}(E) \, \di E = \frac {\abs{\Lambda^{\pr\pr}}}{\abs{\Lambda}} K_W \rho_+ \norm{f}_1 ,
\notag
\end{align}
and
\begin{align}
&  \E \abs{ \tr\set{\chi_{\Lambda^{\pr\pr}} f_{\Lambda}(H_{\bom})\chi_{\Lambda^{\pr\pr}} }}\le  \abs{\Lambda^{\pr\pr}}
 \E \tr\set{ \chi_0\abs{ f_{\Lambda}}(H_{\bom}^{(\Lambda)})\chi_0} = \abs{\Lambda^{\pr\pr}} N(\abs{ f_{\Lambda}})\\
 & \qquad \quad  \le \abs{\Lambda^{\pr\pr}}  K_W \rho_+
 \int_\R  \abs{ f_{\Lambda}}(E) \, \di E = \frac {\abs{\Lambda^{\pr\pr}}}{\abs{\Lambda}} K_W \rho_+ \norm{f}_1 .
\notag
\end{align}
Since $\frac {\abs{\Lambda^{\pr\pr}}}{\abs{\Lambda}}\approx  \frac 1 {\sqrt{L}}$, the term in \eq{boundaryterm2} goes to $0$ as $L \to \infty$. 

To finish the proof of \eq{limL3} , we need to show that the term in  \eq{mainterm2} also goes to $0$ as $L \to \infty$.  As before, we use the  Helffer-Sj\"ostrand formula. We have, taking 
 $\phi_0 \in C_{K}^\infty(\Lambda_{L- 10d}(0))$, such that $0\le \phi_0 \le1$ and $\phi_0 \chi_{\Lambda_{L- 20 d}(0)} = \chi_{\Lambda_{L- 20 d}(0)} $, that
\begin{align}\label{TLambda2}
S_{\bom}^{(\Lambda)}:=& \chi_{\Lambda^{\pr}} f_{\Lambda}(H_{\bom})\chi_{\Lambda^{\pr}}  -   \chi_{\Lambda^{\pr}}  f_{\Lambda}(H_{\bom}^{(\Lambda}) \chi_{\Lambda^{\pr}} \\ \notag
 = &\int_{\R^{2}} \!\di \tilde{f_{\Lambda}}(z) \set{\chi_{\Lambda^{\pr}} R_{\bom}(z)\chi_{\Lambda^{\pr}} -\chi_{\Lambda^{\pr}}  R_{\bom}^{(\Lambda)}(z)\chi_{\Lambda^{\pr}} } \\ \notag
 = & \int_{\R^{2}} \!\di \tilde{f_{\Lambda}}(z) \set{\chi_{\Lambda^{\pr}}R_{\bom} (z)\phi_0 \chi_{\Lambda^{\pr}}-\chi_{\Lambda^{\pr}}\phi_0  R_{\bom}^{(\Lambda)}(z) \chi_{\Lambda^{\pr}} }   \\ 
= & \int_{\R^{2}} \!\di \tilde{f_{\Lambda}}(z) \set{\chi_{\Lambda^{\pr}} R_{\bom}(z)W(\phi_0)R_{\bom}^{(\Lambda)}(z)\chi_{\Lambda^{\pr}}}  .\label{intHS2}
\end{align}
Proceeding as in \eq{expres}-\eq{breakest}, we get
\begin{align}\label{breakest2}
\norm{\chi_{\Lambda^{\pr}} R_{\bom}(z)W(\phi_0)R_{\bom}^{(\Lambda)}(z)\chi_{\Lambda^{\pr}} }_1  \le C_4 C^\pr_f L^{p_d} \abs{\Im z}^{-2 p_d}\norm{\chi_{\nabla \phi_{0}}R_{\bom}^{(\Lambda)}(z)\chi_{\Lambda^{\pr}}}.  
\end{align}

Recall  $\overline{\cI_1}\subset \Xi^{\text{CL}}$.  The term in \eq{mainterm2} is $ \E\set{S_{\bom}^{(\Lambda)}}$, with $S_{\bom}^{(\Lambda)}$ as in \eq{TLambda2}.  It follows  from \eq{intHS2} and \eq{breakest2} that for large $L$,
\begin{align}\notag
& \E\set{S_{\bom}^{(\Lambda)}}\le  C_4 C^\pr_f L^{p_d} \int_{\R^{2}} \!\abs{\di \tilde{f_\Lambda}(z) }\, \abs{\Im z}^{-2 p_d}
\E\set{\norm{\chi_{\nabla \phi_{0}}R_{\bom}^{(\Lambda)}(z)\chi_{\Lambda^{\pr}}}}\\
& \quad \le  M_L  C_4C^\pr_f L^{p_d} \int_{\R^{2}} \!\abs{\di \tilde{f_\Lambda}(z) }\, \abs{\Im z}^{-2 p_d- \frac 4 5}
\E\set{\norm{\chi_{\nabla \phi_{0}}R_{\bom}^{(\Lambda)}(z)\chi_{\Lambda^{\pr}}}^{\frac 1 5} } \\ \notag
& \quad\le   C_4C^\pr_f L^{p_d + 2d}(\rho_+  + \sqrt{\rho_+}) \int_{\R^{2}} \!\abs{\di \tilde{f_\Lambda}(z) }\, \abs{\Im z}^{-2 p_d- \frac 4 5}  \mathrm{e}^{-L^{\frac 1 4}}\\\notag
&\quad \le L^{p_d +2 d} \mathrm{e}^{-L^{\frac 1 4}}c_{2 p_d+ \frac 4 5} C_4 C^\pr_f(\rho_+  + \sqrt{\rho_+}) \hnorm{f_\Lambda}_{2p_d +2}\\\notag
& \quad\le L^{d^2 + 5d}   \mathrm{e}^{-L^{\frac 1 4}}c_{2 p_d+ \frac 4 5} C_{f,E_0,d} (\rho_+  + \sqrt{\rho_+})\to 0 \; \text{as $L\to \infty$}\notag
\end{align}
where we used \eq{eqMSAexp} and \eq{HShigherorder}.

Thus \eq{limL3}, and the lemma, is proven.
\end{proof}

Given point processes $\set{\zeta_n}_{n\in \N}$ and  $\zeta$ on 
$\R$, we let $\zeta_n \Rightarrow \zeta$  denote the weak convergence of $\zeta_n$ to $\zeta$ as $n \to \infty$. We recall \cite[Proposition~9.1.VII]{DV} that $\zeta_n \Rightarrow \zeta$ if and only if
\beq \label{weakconv}
\lim_{n\to \infty} \E \, \e^{-\zeta_n(f)}= \E \, \e^{-\zeta(f)} \quad \text{for all}\quad f \in C_{K,+}(\R).
\eeq

The following lemma shows that it suffices to prove that $\widetilde{\xi}_{\bom}^{(\Lambda)} \Rightarrow \xi$ to prove Theorem~\ref{thmPoisson}(b).

\begin{lemma}\label{lemweakconv}
 ${\xi}_{\bom}^{(\Lambda)} \Rightarrow \xi$ if and only if  \  $\widetilde{\xi}_{\bom}^{(\Lambda)} \Rightarrow \xi$. 
\end{lemma}

\begin{proof} 
If $\zeta_i$, $i=1,2$, are point processes on $\R$, defined on the same probability space, 
 we have,  for all $ f \in C_{K,+}(\R)$,
\beq \label{Laplace}
\abs{ \E \, \e^{-\zeta_1(f)}-  \E \, \e^{-\zeta_2(f)}}\le \E \, \abs{\zeta_1(f)-\zeta_2 (f)}.
\eeq
The lemma follows immediately from \eq{weakconv}, \eq{Laplace}, and Lemma~\ref{lemEconv}.
\end{proof}

We are now ready to prove Theorem~\ref{thmPoisson}(a).   In view of Lemma~\ref{lemweakconv}, it suffices to prove that  $\widetilde{\xi}_{\bom}^{(\Lambda)} \Rightarrow \xi$.  By standard results from the theory of point processes (cf.\  \cite[Theorem~9.2.V 	and  subsequent remark]{DV}; see also \cite[Theorem~2.3]{Kr}), this is equivalent to verifying the following three conditions  for all bounded intervals $I$ (recall $\Lambda=\Lambda_L(0)$):
\begin{gather}\label{cond1}
\lim_{L \to\infty} \max_{m=1,2,\ldots,M_L} \P\{ \xi_{\bom}^{(\Lambda,m)}(I) \ge 1\} = 0,\\\label{cond2}
\lim_{L \to\infty} \sum_{m=1}^{M_L}  \P\{ \xi_{\bom}^{(\Lambda,m)}(I) \ge 1\}  = n(\cE) \abs{I},\\
\lim_{L \to\infty} \sum_{m=1}^{M_L}  \P\{ \xi_{\bom}^{(\Lambda,m)}(I) \ge 2\}  = 0 .\label{cond3}
\end{gather}
Since $\P\{ \xi_{\bom}^{(\Lambda,m)}(I) \ge 1\} \le \E \set{ \xi_{\bom}^{(\Lambda,m)}(I)}$,
 \eq{cond1}  follows  immediately from \eq{int1}. In addition, it follows from the definition \eq{defxim} and the estimate \eq{Minami}, that for all $\Lambda$ with    $\abs{\Lambda}\ge c_I$ we have 
\beq \label{Minamiconseq}
 \P\{ \xi_{\bom}^{(\Lambda,m)}(I) \ge 2\} \le
  \tfrac 1 2 \E \set{\pa{ \xi_{\bom}^{(\Lambda,m)}(I)}\pa{ \xi_{\bom}^{(\Lambda,m)}(I)-1}}\le  
    \tfrac 1 2 K_M \pa{\rho_+ \abs{I} M_L^{-1} }^2,
\eeq 
 so \eq{cond3} follows.
 
Thus Theorem~\ref{thmPoisson}(a) is proved if we  verify condition \eq{cond2}. To do so, we first notice that
\beq
\E \set{ \xi_{\bom}^{(\Lambda,m)}(I)} = \sum_{k=1}^\infty \P\{ \xi_{\bom}^{(\Lambda,m)}(I) \ge k\},
\eeq
and, as in \cite{Kr},
\beq \begin{split}
&\sum_{k=2}^\infty \P\{ \xi_{\bom}^{(\Lambda,m)}(I) \ge k\}= \sum_{k=2}^\infty (k-1)\P\{ \xi_{\bom}^{(\Lambda,m)}(I) =k\}\\
&  \qquad \le  \sum_{k=2}^\infty k(k-1)\P\{ \xi_{\bom}^{(\Lambda,m)}(I) =k\} = \E \set{\pa{ \xi_{\bom}^{(\Lambda,m)}(I)}\pa{ \xi_{\bom}^{(\Lambda,m)}(I)-1}}.
\end{split}\eeq
It thus follows, as in \eq{Minamiconseq}, that
\beq
0\le  \E \set{ \widetilde{\xi}_{\bom}^{(\Lambda)} (I)}\! -\! \sum_{m=1}^{M_L}  \P\{ \xi_{\bom}^{(\Lambda,m)}(I) \ge 1\} \le  M_L K_M \pa{\rho_+ \abs{I} M_L^{-1} }^2\!\! \to 0 \, \text{as} \, L \to \infty.
\eeq
We conclude that \eq{cond2} is equivalent to
\beq
\lim_{L \to\infty}   \E \set{ \widetilde{\xi}_{\bom}^{(\Lambda)} (I)} = n(\cE) \abs{I},
\eeq
and hence, by Lemma~\ref{lemEconv}, equivalent to 
 \beq \label{cond2rev}
\lim_{L \to\infty}   \E \set{{\theta}_{\bom}^{(\Lambda)}(I)} = n(\cE) \abs{I}.
\eeq
But it follows from \eq{thetaint} that, for  all $\Lambda$ such that   $\abs{\Lambda}\ge c_I$
\beq
  \E \set{{\theta}_{\bom}^{(\Lambda)}(I)}= \abs{\Lambda} \eta(\cE + \abs{\Lambda}^{-1} I)=
  \abs{\Lambda}  \int_{\cE + \abs{\Lambda}^{-1} I} \, n(E)\,  \di E.
\eeq
Since by our hypothesis  $\cE$ is a Lebesgue point of the locally integrable function $n(E)$
(cf.\  \cite[Definition~25.13]{Y}), and the sets $\cE + \abs{\Lambda}^{-1} I$ shrink nicely to $\cE$ as $L\to \infty$ (cf.\  \cite[Definition~25.16]{Y}), we can use the Lebesgue Differentiation Theorem  (cf.\  \cite[Theorem~25.17]{Y}) to conclude
\beq
\lim_{L\to \infty}  \abs{\Lambda}  \int_{\cE + \abs{\Lambda}^{-1} I} n(E) \, \di E= n(\cE) \abs{I}.
\eeq

Thus \eq{cond2rev}, and hence  \eq{cond2}, is proven, completing the proof of    Theorem~\ref{thmPoisson}(a).

\section{Simplicity of eigenvalues}\label{secsimple}

We prove  Theorem~\ref{thmPoisson}(b) proceeding as in  \cite{KM}.  Let $H_\bom$ be an Anderson Hamiltonian, and let  $\mathcal{I} $ be  an open interval
such  that for large boxes $\Lambda$ the estimate \eq{Minami} holds for any
interval $I \subset \mathcal{I}$ with $\abs{I}\le \delta_0$, for some $\delta_0 >0$, with some constant $K_M$.
   We call
 $\vphi \in \L^2(\R^d)$
 fast decaying if it has $\beta$-decay for some $\beta > \frac {5d} 2$, which in the continuum means that
 $\norm{\chi\up{1}_x\vphi} 
\le C_\vphi \la x\ra^{-\beta}$ for some constant  $C_\vphi $,
where  $ \langle x \rangle :   = \sqrt{1+|x|^2} $. 
We will show that, with probability one, $H_\bom$  cannot  have an eigenvalue in  $\mathcal{I} $  with $2$  linearly independent
fast decaying eigenfunctions.

Let $I\subset \mathcal{I}$ be a  closed interval,  $q > 2d$, $L \in 2\N$ large, $\Lambda_L=\Lambda_L(0)$.
  We   cover the interval $I$ by $2 \left (\left[\frac {L^q} 2 |I|\right]+1\right)\le 
{L^q}  |I| +2 $
intervals of length $2L^{-q} $, in such a way that 
any subinterval $J \subset I $
with length $|J| \le L^{-q} $ will be contained in one of these intervals. 
($[x]$ denotes the largest integer $\le x$.)  Let
$\mathcal{B}_{L,I,q}$ denote the complement to the  event   that 
 $\tr P\up{\Lambda_L}_\bom(J) \le 1$ for all subintervals $J \subset I $
with length $|J| \le  L^{-q} $.
The probability of  $\mathcal{B}_{L,I,q}$
can be estimated, using
\eq{Minami} and 
\beq \label{2eigenvalues}
\P \set{\tr P_\bom^{(\Lambda)}(I)\ge 2}\le  \tfrac 1 2 \E \set{ \pa{\tr P_\bom^{(\Lambda)}(I)}\pa{\tr P_\bom^{(\Lambda)}(I)-1 } },
\eeq
by
\begin{equation}
\P\{\mathcal{B}_{L,I,q}\}\le \tfrac 1 2  K_M \rho_+^2 ( L^q |I|+2) (2L^{-q})^2 L^{2d}
\le   2  K_M \rho_+^2  (|I| +1) L^{-q+ 2d}.
\end{equation}
Thus, taking scales $L_k=2^k$, $k=1,2,\ldots$,  it follows  from the Borel-Cantelli Lemma
that, with probability one,   the event   $\mathcal{B}_{L_k,I,q}$
eventually does not occur.

Let $\bom$ be in the set of probability one for which we have pure point spectrum with exponentially decaying eigenfunctions in the region of complete localization  $\Xi^{\text{CL}}$.  Suppose there exists  $E\in \mathcal{I}\cap \Xi^{\text{CL}}$ which  is an eigenvalue of
$H_\bom$ with   $2$ linearly independent
 eigenfunctions.   In particular these eigenfunctions decay exponentially,    so,
if we fix  $\beta>\frac {5d} 2$, they both have  $\beta$-decay.   Pick an open interval $I\ni E$, such that $ \bar{I} \subset  \mathcal{I}\cap \Xi^{\text{CL}}$.
 \cite[Lemma~1]{KM} can be adapted to the continuum by using smooth functions to localize the eigenfunctions in finite boxes.    It then follows that for   $L$ large enough the
finite volume operator  $H_{\bom}\up{ \Lambda_L}$  has at least $2$ eigenvalues in the
interval $J_{E,L}=[E-\eps_L,E+\eps_L]$, where $\eps_L= C L^{-\beta +\frac d 2}$ for an
appropriate  constant $C$ independent of $L$.  Since $\beta>\frac {5d} 2$  there exists $q > 2d$
such that  $\beta - \frac d 2 > q$, and hence  $\eps_L < L^{-q}$ for all large $L$.  But  with
probability one this is impossible since  the event   $\mathcal{B}_{L_k,\bar{I},q}$
 does not occur for large $L_k$.

Theorem~\ref{thmPoisson}(b) is proven.

\appendix

\section{The region of complete localization}\label{appMSA}

In this appendix we discuss  localization for  an   Anderson Hamiltonian
$H_\bom$. Localization is most commonly taken to be 
\emph{Anderson localization}:  pure point
spectrum with exponentially decaying eigenstates with probability one.
It is also natural to consider   \emph{dynamical localization}:
the moments of a
 wave packet, initially localized both in
space and in energy,  should remain 
uniformly bounded under  time evolution.  For the multi-dimensional continuum  Anderson Hamiltonian, localization has been  proved by a multiscale analysis  \cite{HM,CH,Klo93,KSS,GdB,DS,GKboot,GKgafa}, and, in the case when we have the covering condition  $\delta_- \ge 1$, also by the fractional moment method \cite{AENSS}.
These methods give more than just  Anderson or dynamical localization, although they imply both.  In the case when both methods are available, i.e.,  $\delta_- \ge 1$, they have the same region of applicability
(see \cite{GKsudec,Kle}).

Thus, following \cite{GKsudec}, we consider
 \emph{the region of complete localization}  $\Xi^{\text{CL}}$ for an   Anderson Hamiltonian
$H_\bom$,
defined as the set of energies  $E \in \R$  where  we have the conclusions of the bootstrap 
multiscale analysis of \cite{GKboot}, ie., as the
 set of $E \in \R$ for which there exists some
open interval $I \ni E$, such that 
 given any $\zeta$, $0<\zeta<1 $, and $\alpha$,
 $ 1<\alpha<\zeta^{-1}$, there is a length scale
$L_0\in 2 \mathbb{N}$
 and a mass $m>0$, so if we
 take   $L_{k+1} \approx L_k^\alpha$, with $L_{k+1} \in 2 \mathbb{N}$,
 $k=0,1,\dots$,
 we have
\begin{equation} \label{MSAest}
\mathbb{P}\,\left\{R\left(m, L_k, I,x,y\right) \right\}\ge 1
-\mathrm{e}^{-L_{k}^\zeta}
\end{equation} for all $k=0,1,\ldots$, and  $x, y \in \mathbb{Z}^d$
with 
$|x-y| > L_k + \varrho $, where $\varrho>0$ is a constant depending only on $\supp u$, and  
\begin{eqnarray}
\lefteqn{R(m,L, I,x,y) =} \label{defsetR} \\ 
&&\{\mbox{$\omega$; for every}
\; E^\prime \in I
\ \mbox{either} \ \Lambda_L(x) \ \mbox{or} \ 
\Lambda_L(y) \  \mbox{is} \  \mbox{$(\omega,m, E^\prime)$-regular} \} \ .
\nonumber
\end{eqnarray} 
  Given $E \in \mathbb{R}$,
$x \in  \mathbb{Z}^d$ and $L \in 6 \mathbb{N}$,  
we say that
 the box 
$\Lambda_L(x) $ is
$(\omega,m, E)$-regular  for a
 given $m>0$ if $E \notin \sigma(H\up{\Lambda_L(x)}_{\bom})$ and 
\begin{equation} 
\| \Gamma\up{L}_x  R\up{\Lambda_L(x)}_{\bom}(E + i \delta)\chi_{\Lambda_{\frac L 3}(x)} 
\| \le {\rm e}^{-m\frac{L}{2}} \quad \text{for all} \quad \delta\in \R,
 \label{regular}
\end{equation}
where $R\up{\Lambda_L(x)}_{\bom}(E+  i\delta) = (H\up{\Lambda_L(x)}_{\bom}-(E+ i\delta))^{-1}$ and $\Gamma\up{L}_x$
denotes the charateristic function of the ``belt" 
$ \overline{\Lambda}_{L-1}(x) \backslash {\Lambda}_{L-3}(x) $.  (See \cite{GKboot, GKduke,GKsudec,Kle}; note that all the proofs work with the definition \eq{regular}, i.e., with the insertion of ``for all $\delta \in \R$".  They also work with the finite volume operators with periodic boundary condition used in this article.)

By construction $ \Xi^{\text{CL}}$ is an open set.   It can be characterized in many different ways  \cite{GKduke,GKsudec}.   For convenience,  our definition includes the complement of the spectrum of $H_\bom$ in the region of complete localization, that is, $\R \setminus \Sigma \subset \Xi^{\text{CL}}$.  The spectral region of complete localization, $\Xi^{\text{CL}}\cap \Sigma $, is called the ``strong insulator region" in
 \cite{GKduke}.)
 If the conditions
for the fractional moment method are satisfied, $\Xi^{\text{CL}}$
 coincides with the set of 
energies   where the fractional moment method can be performed.
(Minami \cite{Mi} proved Poisson statistics for the Anderson model in the region of validity of the fractional moment method, in other words, in the region of complete localization for the Anderson model.)

We use the following estimate.

\begin{proposition}\label{propMSA} Consider a closed bounded interval $I \subset \Xi^{\text{CL}}$.  Then for all  $z \in \C$ with $\Re z \in I$, and boxes $\Lambda=\Lambda_L$, we have, for  $s \in ]0,\frac 14[$ and $\xi \in ]0,1[$,
and $x,y \in \Lambda$ with $\abs{x-y} \ge \pa{\log L }^{\frac 1 \xi  +}$,
\beq \label{eqMSAexp}
\E \set{\norm{\chi\up{1}_x R_\bom^{(\Lambda)}(z)\chi\up{1}_y}^s} \le C_{s,I,\zeta} (\rho_+  + \sqrt{\rho_+}) \mathrm{e}^{-\abs{x-y}^\xi}\
\eeq
for $L \ge L_1(\xi,I,s)$.
\end{proposition}

We will need the following consequence of the Wegner estimate \eq{Wegner}.

\begin{lemma}\label{lemWexp}  Let $I=[c,d]$ be such that  \eq{Wegner} holds for any subinterval of $[c-1,d+1]$ with a constant $K_W $.  Then for any $s \in \left]0,\frac 1 2\right [$,  box $\Lambda$, and $z \in \C$ with $\Re z \in I$, we have
\beq \label{Wexp}
\E \set{\norm{R_\bom^{(\Lambda)}(z)}^s}\le C_s  K_W \rho_+ \abs{\Lambda}.
\eeq
 \end{lemma}
 
 \begin{proof}  Let $\Re z \in I$. It follows from
  \eq{Wegner} that for all $t \ge 1$
  \beq
  \P\set{\norm{R_\bom^{(\Lambda)}(z)} \ge t   } \le \tfrac 2 t K_W \rho_+ \abs{\Lambda}
  \eeq
  Thus
  \beq \begin{split}
  \E \set{\norm{R_\bom^{(\Lambda)}(z)}^s}& = \int_0^\infty  t \,   \P\set{\norm{R_\bom^{(\Lambda)}(z)}^s \ge t }    \di t  \\
 &\le
  1 + \int_1^\infty  t \, \pa{ \tfrac 2 {t^{\frac 1 s}} K_W \rho_+ \abs{\Lambda} }  \di t
 \le 1 +C^\pr_s  K_W \rho_+ \abs{\Lambda} .
 \end{split} \eeq
   \end{proof}
   
   If we have the covering condition $\delta_- \ge 1$, \eq{Wexp} holds without the volume factor in the right hand side \cite{AENSS}.  
   
\begin{proof}[Proof of Proposition~\ref{propMSA}]  
   Given $0<\xi<1$, we pick $\zeta$ such that
$\zeta^2 <\xi<\zeta <1$  (always possible) and set $\alpha = \frac \zeta \xi $,
note $\alpha < \zeta^{-1}$. 
Since $I \subset \Xi^{\text{CL}}$,  there is a scale $L_0\in 2 \mathbb{N}$ and a 
mass $m_\zeta>0$, such that, if
we set   $L_{k+1} \approx L_k^\alpha$, with $L_{k+1} \in 2 \mathbb{N}$,
 $k=0,1,\dots$,
we have the estimate
(\ref{MSAest})
for  $x, y \in \mathbb{Z}^d$ such that $ |x-y|> L_k + \varrho $.

Let us now fix $\Lambda=\Lambda_L$,  $x,y \in \Lambda_L \cap \mathbb{Z}^d$, and pick $k$ such that
 $L_{k+1} +\varrho \geq |x-y|> L_k +\varrho$.  In this case,
if $\bom\in R\left(m_\zeta, L_k, I,x,y\right)$, then for $\Re z \in I $ either 
$\Lambda_{L_k}(x)$ {or}
$\Lambda_{L_k}(y)$ {is} $(\omega,m, \Re z)$-regular; say $\Lambda_{L_k}(x)$ {is} $(\omega,m, \Re z)$-regular.  (Note that we take the boxes of size $L_k$ in the torus $\Lambda$.)  Then, using \cite[Eq.~(2.9)]{GKboot} \and \eq{regular}, 
\begin{align}
\norm{\chi\up{1}_y R_\bom^{(\Lambda)}(z)\chi\up{1}_x}& \le \gamma_I \norm{\Gamma_x\up{L_k} R_\bom^{(\Lambda_{L_k}(x))}(z)\chi\up{1}_x}\norm{\chi\up{1}_y R_\bom^{(\Lambda)}(z)\Gamma_x\up{L_k}}
\\
&  \le \gamma_I  {\rm e}^{-m_\zeta \frac{L_k}{2}}\norm{R_\bom^{(\Lambda)}(z)}. \notag
\end{align}

Thus, with $s \in ]0,\frac 1 4[$, using Lemma~\ref{lemWexp},
\begin{align}\notag
&\E \set{\norm{\chi\up{1}_yR_\bom^{(\Lambda)}(z)\chi\up{1}_x}^s\!\! ; \, \bom\in R\left(m_\zeta, L_k, I,x,y\right)\! }
\le  \gamma_I^s  {\rm e}^{-s m_\zeta \frac{L_k}{2}} \E \set{\norm{R_\bom^{(\Lambda)}(z)}^s} \\
& \qquad \qquad \qquad  \le C_s  K_W \rho_+ \abs{\Lambda} \gamma_I^s  {\rm e}^{-s m_\zeta \frac{L_k}{2}}
\le C_{s,I }\rho_+  \abs{\Lambda}  {\rm e}^{-s m_\zeta \frac{L_k}{2}},
\end{align}
and
\begin{align}\notag
&\E \set{\norm{\chi\up{1}_yR_\bom^{(\Lambda)}(z)\chi\up{1}_x}^s; \ \bom\notin R\left(m_\zeta, L_k, I,x,y\right) }\\
& \qquad \quad \le  \pa{\E \set{\norm{R_\bom^{(\Lambda)}(z)}^{2s}}}^{\frac 1 2}\pa{\P\set{\bom\notin R\left(m_\zeta, L_k, I,x,y\right)}}^{\frac 1 2}\\ \notag
& \qquad \quad \le   \pa{C_{2s}  K_W \rho_+ \abs{\Lambda}}^{\frac 1 2}\mathrm{e}^{-\frac 1 2L_{k}^\zeta}
\le C_{s,I }^\pr \pa{ \rho_+ \abs{\Lambda}}^{\frac 1 2}\mathrm{e}^{-\frac 1 2L_{k}^\zeta}.
\end{align}

It follows, that for $L_k$ sufficiently large, that is, $\abs{x-y}$ large, we have
\begin{align}
&\E \set{\norm{\chi\up{1}_yR_\bom^{(\Lambda)}(z)\chi\up{1}_x}^s}\le C_{s,I,\zeta}(\rho_+  + \sqrt{\rho_+})\abs{\Lambda} \mathrm{e}^{-\frac 1 2 L_{k}^\zeta} \\
& \quad \le  C_{s,I,\zeta}(\rho_+  + \sqrt{\rho_+})\abs{\Lambda} \mathrm{e}^{-\frac 1 2 L_{k+1}^\xi}
\le   C_{s,I,\zeta}^\pr (\rho_+  + \sqrt{\rho_+})\abs{\Lambda} \mathrm{e}^{-\frac 1 2 \abs{x-y}^\xi},\notag
\end{align}
so \eq{eqMSAexp} follows for $\abs{x-y} \ge  \pa{\log L }^{\frac 1 \xi  +} $ (with a slightly smaller $\xi$).
     \end{proof}

\section{A convexity inequality for traces}

The following inequality  was used in \cite[Proof of Proposition~4.5]{CH}.

\begin{lemma}\label{lemconvexJ}
Let $H_1$ and $ H_2$ be two self-adjoint operators on a Hilbert space $\H$, such that $H_1$ is diagonalizable and $H_1 \ge H_2$.  Let $f$ and $ g$ be bounded Borel functions on some open interval $I \supset \sigma(H_1)$, such that  $g$ is real-valued, nonincreasing, and convex on $I$.  Then
\beq \label{convexJ}
\tr \set{\bar{f}(H_1) g (H_1) f(H_1)  } \le\tr \set{\bar{f}(H_1) g (H_2) f(H_1)  }.
\eeq
\end{lemma}

\begin{proof}
Let $\varphi \in \H$, $ \norm{ \varphi}=1$, be an eigenvector of $H_1$ with eigenvalue $\lambda$, that is, $H_1  \varphi=\lambda  \varphi$.   Then
\begin{align}\label{convexJ1}
&\langle \varphi, \bar{f}(H_1) g (H_1) f(H_1) \varphi\rangle= \bar{f}(\lambda) g (\lambda) f(\lambda) = \bar{f}(\lambda) g\pa{ \langle \varphi, H_1\varphi\rangle} f(\lambda)\\
& \quad \le   \bar{f}(\lambda) g\pa{ \langle \varphi, H_2\varphi\rangle}f(\lambda)
\le   \bar{f}(\lambda)  \langle \varphi, g(H_2)\varphi\rangle f(\lambda)
=   \langle \varphi, \bar{f}(H_1) g (H_2) f(H_1) \varphi\rangle,
\notag
\end{align}
where the first inequality follows from  $g$ nonincreasing and $H_1\ge H_2$, and the second inequality used the convexity of the function $g$, Jensen's inequality (cf.\  \cite[Theorem~14.16]{Y}), and the spectral theorem. 

Since $H_1$ is diagonalizable, the inequality \eq{convexJ} follows by expanding the trace on an orthonormal basis of eigenvalues fot $H_1$ and using \eq{convexJ1} for each term.
\end{proof}




\begin{thebibliography}{AENSS}

\bibitem[AENSS]{AENSS}  Aizenman, M., Elgart, A., Naboko, S.,   
 Schenker, J.,  Stolz, G.: Moment analysis for localization in random Schr\"odinger 
operators. Inv. Math. \textbf{163}, 343-413 (2006)

\bibitem[B]{B} Bellissard, J.: Random matrix theory and the Anderson model.  J. Statist. Phys.~{\bf 116},  739-754 (2004)

\bibitem[BHS]{BHS} Bellissard, J., Hislop, P., Stolz, G.: Correlation estimates in the
Anderson model.  J. Stat. Phys.  {\bf 129},  649-662 (2007)

\bibitem[CKM]{CKM}  Carmona, R.,  Klein, A.,   Martinelli, F.: {Anderson
localization for  Bernoulli and other singular potentials}.  Commun.
Math. Phys. {\bf 108}, 41-66 (1987)


\bibitem[CL]{CL}  Carmona, R,  Lacroix, J.: {\em Spectral theory of random
Schr\"odinger operators}. Boston: Bir\-kha\"user, 1990

\bibitem[CoGK]{CGK}  Combes,  J.M., Germinet, F.,  Klein, A.: Generalized eigenvalue-counting estimates for the Anderson model.  J. Stat. Phys. \textbf{135},   201-216 (2009)


\bibitem[CoH]{CH}  Combes,  J.M.,   Hislop, P.D.: {Localization for some
continuous, random Hamiltonians in d-dimension}. J. Funct. Anal. \textbf{124},
149-180 (1994)



\bibitem[CoHK1]{CHK1} Combes,  J.M.,   Hislop, P.D.,  Klopp, F.: 
 H\"older continuity of the integrated density of states for some random
 operators at all energies.  IMRN \textbf{4}, 179-209 (2003)

\bibitem[CoHK2]{CHK2} Combes,  J.M.,   Hislop, P.D.,  Klopp, F.: 
 An optimal Wegner estimate and its application to the global continuity of the integrated density of states for random Schr\"odinger operators.  Duke Math. J.  {\bf 140},  469-498  (2007)


\bibitem[CoHK3]{CHK3} Combes,  J.M.,   Hislop, P.D.,  Klopp, F.: 
 Some new estimates on the spectral shift function associated with
 random Schr\"odinger operators.
 \emph{Probability and mathematical physics} 
 85-95, CRM Proc. Lecture Notes \textbf{42}, Amer. Math. Soc., Providence, RI,  2007 

\bibitem[DV]{DV}  Daley, D. J.,  Vere-Jones, D.:   \emph{An introduction to the theory of point processes}.
Springer Series in Statistics. Springer-Verlag, New York,  1988

\bibitem[DaS]{DS}  Damanik, D.,  Stollmann, P.: {Multi-scale analysis implies
strong dynamical localization}.  Geom. Funct. Anal. {\bf 11}, 11-29 (2001)

\bibitem[DiPS]{DPS} Disertori, M., Pinson, H., Spencer, T.: Density of states for random band matrices. Commun. Math. Phys. 232, 83-124 (2002)

\bibitem[ESY1]{ESY1}  Erdos, L., Schlein, B., Yau, H.-T.: Semicircle law on short scales and delocalization of eigenvectors for Wigner random matrices.  Ann. Probab. To appear

\bibitem[ESY2]{ESY2}  Erdos, L., Schlein, B., Yau, H.-T.: Local semicircle law and complete delocalization for Wigner random matrices. Commun. Math. Phys. To appear


\bibitem[FK]{FKac} Figotin, A.,  Klein, A.: {Localization of classical
waves I: Acoustic waves}.  Commun. Math. Phys. {\bf 180}, 439-482
(1996)

\bibitem[FrS]{FS}  Fr\"ohlich, J.,  Spencer, T.: {Absence of diffusion with
Anderson tight binding model for large disorder or low energy}. Commun.
Math. Phys. {\bf 88}, 151-184 (1983)

\bibitem[GD]{GdB} Germinet, F., De Bi\`evre, S.: {Dynamical localization for
discrete and continuous random Schr\"odinger operators}.  Commun.
Math. Phys. {\bf 194}, 323-341 (1998)

\bibitem[GK1]{GKboot} Germinet, F.,  Klein, A.: {Bootstrap multiscale analysis and
localization in random media}.  Commun. Math. Phys. \textbf{222},
415-448 (2001)

\bibitem[GK2]{GKdecay}  Germinet, F,  Klein, A.: Operator kernel estimates for  functions of  generalized Schr\"odinger operators.
Proc. Amer. Math. Soc. \textbf{131},  911-920  (2003)

\bibitem[GK3]{GKgafa} Germinet, F.,  Klein, A.: {Explicit finite volume criteria for
localization in random media and applications}.  Geom. Funct. Anal. \textbf{13}, 
1201-1238 (2003)


\bibitem[GK4]{GKduke} Germinet, F.,  Klein, A.: {A characterization of the
Anderson metal-insulator transport transition}.
Duke Math. J. \textbf{124}, 309-351 (2004)

\bibitem[GK5]{GKsudec} Germinet, F., Klein, A.: New characterizations of the region of complete localization for random Schr\"odinger operators.  J. Stat. Phys. \textbf{122}, 73-94 (2006)

\bibitem[GoMP]{GMP}  Gol'dsheid, Ya., Molchanov, S., Pastur, L.:  Pure point
spectrum of stochastic one dimensional Schr\"odinger operators.
Funct. Anal. Appl. {\bf 11}, 1-10 (1977)

\bibitem[GrV]{GV} Graf, G.-M., Vaghi, A.: A remark on an estimate by Minami,  Lett. Math.
Phys.  {\bf 79},  17-22 (2007)

\bibitem[J1]{J1} Johansson, K.: On fluctuations of eigenvalues of random Hermitian matrices. Duke Math. J.~{\bf 91}, 151-204 (1998) 

\bibitem[J2]{J2} Johansson, K.: Universality of the local spacing distribution in certain ensembles of Hermitian Wigner
matrices. Comm. Math. Phys.~{\bf 215}, 683–705 (2001).

\bibitem[HM]{HM}   Holden, H., Martinelli, F.:  On absence of diffusion near
the bottom of the spectrum for a random Schr\"odinger operator.
 Commun. Math. Phys. {\bf 93}, 197-217 (1984)



\bibitem[HuS]{HuS}
{Hunziker}, W., {Sigal}, I. M.: 
 {Time-dependent scattering theory of $N$-body quantum systems}.
 {Rev. Math. Phys.}  \textbf{12},  {1033-1084}  (2000)
 
 
\bibitem[KS]{KS} Killip, S., Stoiciu, M.: Eigenvalue Statistics for CMV Matrices: From Poisson to Clock via Circular Beta Ensembles. Duke Math. J., to appear

 
\bibitem[Ki]{Ki} Kirsch, W.: An invitation to random Schr\"odinger operators.  In  \emph{Random SchrodingerOperators}.  Panoramas et Synth\`{e}ses \textbf{25},   1-119,  
 Soci\'{e}t\'{e}  Math\'{e}matique de France, Paris 2008


\bibitem[KiM]{KiM}  Kirsch, W.,   Martinelli, F.: {On the ergodic properties of 
the spectrum of general random operators}.  J. Reine Angew. Math. {\bf
334},  141-156  (1982)


 
 \bibitem[KiSS]{KSS} Kirsch, W.,   Stollmann,  P, Stolz, G.:
{Localization for random perturbations of periodic Schr\"odinger
 operators}.  Random Oper. Stochastic Equations {\bf 6},
 241-268 (1998) 


\bibitem[Kl]{Kle} Klein, A.:   Multiscale analysis and localization of random operators.
 In \emph{Random SchrodingerOperators}.  Panoramas et Synth\`{e}ses \textbf{25},   121-159,  
 Soci\'{e}t\'{e}  Math\'{e}matique de France, Paris 2008





\bibitem[KlK]{KKdef} Klein, A.,  Koines, A.: {A general framework for
localization  of classical waves: I. Inhomogeneous  media and defect
eigenmodes}. Math. Phys. Anal. Geom. \textbf{4}, 97-130 (2001)





\bibitem[KlKS]{KKS} Klein, A.,  Koines, A., Seifert,  M.: {Generalized
eigenfunctions for waves in inhomogeneous media}. 
 J. Funct. Anal. \textbf{190}, 255-291 (2002)
 
 \bibitem[KlM]{KM} Klein. A., Molchanov, S.: Simplicity of eigenvalues in the Anderson
model.  J. Stat. Phys.~{\bf 122}  , 95-99 (2006)

\bibitem[Klo]{Klo93} Klopp,  F.: {Localization for continuous random
Schr\"odinger operators}. Commun. Math. Phys. {\bf 167}, 
553-569 (1995)

 
\bibitem[Kr]{Kr} Kritchevski, E.: Poisson statistics of eigenvalues in the hierarchical Anderson model. Annales Henri Poincar\'e {\bf 9}, 685-709 (2008)
 
\bibitem[M]{Mi} Minami, N.:  Local fluctuation of the spectrum of a multidimensional Anderson tight binding model.
 Comm. Math. Phys.  \textbf{177}, 709-725  (1996)
 
\bibitem[Mo1]{Mo1} Molchanov, S.: Structure of the eigenfunctions of one-dimensional unordered structures. (Russian) Izv. Akad. Nauk SSSR Ser. Mat. \textbf{42}, 70-103 (1978).
English translation: Math. USSR-Izv. \textbf{12}, 69-101 (1978).


\bibitem[Mo2]{Mo} Molchanov, S.: The local structure of the spectrum of the one-dimensional
 Schršdinger operator.
 Comm. Math. Phys.   \textbf{78}  429-446, (1981)

  \bibitem[PF]{PF}  Pastur, L.,  Figotin, A.: {\em Spectra of Random and
Almost-Periodic Operators}.  Heidelberg: Springer-Verlag, 1992

\bibitem[SS]{SS} Schenker, J.,  Schulz-Baldes, H.: Gaussian fluctuations for random matrices with correlated entries.  Int. Math. Res. Not. IMRN  2007,  no. 15, Art. ID rnm047, 36 pp.

\bibitem[Si1]{Si} Simon, B.:   Cyclic vectors in the Anderson model. Special issue
 dedicated to Elliott H. Lieb.  Rev. Math. Phys.  {\bf 6}, 1183-1185   (1994)
 
 \bibitem[Si2]{Si2} Simon, B.:  Spectral averaging and the Krein spectral shift.
 Proc. Amer. Math. Soc.~{\bf 126}, 1409-1413 (1998) 

 \bibitem[St1]{St1} Stoiciu, M.: The statistical distribution of the zeros of random paraorthogonal polynomials on the unit circle. J. Approx. Theory~{\bf 139}, 29-64 (2006) 

 \bibitem[St2]{St2} Stoiciu, M.: Poisson statistics for eigenvalues: from random Schrödinger operators to random CMV matrices. Probability and mathematical physics, 465--475, CRM Proc. Lecture Notes, 42, Amer. Math. Soc., Providence, RI, 2007
 
\bibitem[W]{W} Wegner, F.: Bounds on the density of states in disordered systems, {Z.
Phys. B}{\bf 44} {9-15} (1981)


\bibitem[Y]{Y} Yeh, J.:  \emph{Real analysis.
Theory of measure and integration.}
Second edition.
World Scientific, Hackensack, NJ,  2006



\end{thebibliography}
\end{document}